\documentclass[reqno]{amsart}

\usepackage{a4,amsmath,amssymb,amsthm,slashed}
\usepackage[top=1in, bottom=1in, left=1.25in, right=1.25in]{geometry}
\usepackage{xfrac}
\usepackage{latexsym}
\usepackage{mathrsfs}
\usepackage[numbers]{natbib}
\usepackage{leftidx}
\usepackage{inputenc}

\usepackage{scrextend}

\usepackage{calc}
\usepackage{enumitem}

\usepackage{amsaddr}

\usepackage[T1]{fontenc}

\usepackage{leftidx}
\usepackage{tikz}

\usepackage{hyperref}

\usepackage{color}

\newcommand{\Y}[1]{\textcolor{yellow}{#1}}

\usepackage{url}

\numberwithin{equation}{section}

\usepackage[mathscr]{eucal}

\usepackage{graphicx}
\graphicspath{pic/}

\usepackage{hyperref}
\newcounter{mnotecount}[section]

\newcommand{\half}{\frac{1}{2}}

\newcommand{\eps}{\epsilon}

\newcommand{\veps}{\varepsilon}

\newcommand{\R}{r^2+a^2}

\newcommand{\di}{\mathrm{d}} 

\newcommand{\Horizon}{\mathcal{H}_+}

\newcommand{\CHorizon}{\mathcal{CH}_+}

\newcommand{\Sphere}{\mathbb{S}^2}

\providecommand{\abs}[1]{\lvert#1\rvert}

\makeatletter
  \def\moverlay{\mathpalette\mov@rlay}
  \def\mov@rlay#1#2{\leavevmode\vtop{%
     \baselineskip\z@skip \lineskiplimit-\maxdimen
     \ialign{\hfil$#1##$\hfil\cr#2\crcr}}}
\makeatother

\newcommand{\reg}{\mathbf{s}}

\providecommand{\ellmode}[2]{{#1}_{#2}}

\newcommand{\PJ}{\mathbb{P}}

\newcommand{\Naturals}{\mathbb{N}}

\newcommand{\rb}{\rho}

\def\tb{\tau}

\def\Sigmainit{\Sigma_{\text{init}}}

\newcommand{\bsub}{\begin{subequations}}
\newcommand{\esub}{\end{subequations}}

\def\ptr{\partial}
\def\th{\theta}

\newcommand{\MM}{\mathcal{M}}

\def\HH{\mathcal{H}}
\def\CHH{\mathcal{CH}}

\newcommand{\I}{\mathrm{\mathbf{I}}}
\newcommand{\II}{\mathrm{\mathbf{II}}}
\newcommand{\IIf}{\mathrm{\mathbf{II}_\Gamma}}
\newcommand{\IIfp}{\mathrm{\mathbf{II}'_{\Gamma'}}}

\newcommand{\rred}{r_{\textrm{red}}}
\newcommand{\rblue}{r_{\mathfrak{b}}}

\newcommand{\ub}{\underline{u}}

\def\rmod{r_{\text{mod}}}

\newcommand{\ubin}{\ub_{\textrm{in}}}
\newcommand{\uout}{u_{\textrm{out}}}

\def\phiin{\phi_{\textrm{in}}}
\def\phiout{\phi_{\textrm{out}}}

\newcommand{\rin}{r_{\textrm{in}}}
\newcommand{\rout}{r_{\textrm{out}}}

\newcommand{\thin}{\th_{\textrm{in}}}
\newcommand{\thout}{\th_{\textrm{out}}}

\newcommand{\tub}{\underline{w}}
\newcommand{\tu}{w}

\newcommand{\Cu}{\mathcal{C}}

\newcommand{\pu}{\partial_u}
\newcommand{\pv}{\partial_{\ub}}

\newcommand{\Lxi}{T}
\newcommand{\Leta}{\Phi}
\newcommand{\Th}{\Theta}
\newcommand{\Carter}{\mathbf{C}}

\newcommand{\psilarge}{\psi_{\ell\geq 1}}

\newcommand{\Ckpsilarge}{(\Carter^{k_2}\psi)_{\ell\geq 1}}

\newcommand{\fv}{\mathfrak{F}_{\ub}}
\newcommand{\fu}{\mathfrak{F}_u}

\newcommand{\Entub}{\mathbb{T}_{\tub}[\psi] }

\newcommand{\Enr}{\mathbb{T}_r[\psi] }

\newcommand{\Enn}{\mathbb{E}}

\def\Enw{\tilde{\mathbb{E}}_{deg}}
\def\Ende{\bar{\mathbb{E}}_{deg}}
\def\Endeh{\hat{\mathbb{E}}_{deg}}
\def\Endeg{\mathbb{E}_{deg}}

\newcommand{\psiz}{\psi_{\ell=0}}

\newcommand{\psih}{\psi_{\ell\geq 1}}

\newcommand{\lmu}{|\log(-\mu)|}

\def\dSp{\di\sigma_{\mathbb{S}^2}}
\def\Sp{\Sphere_{u,\ub}}
\def\Y{\hat{e}_3}

\def\Cur{\gamma_{\tub_1,(\phiin)_1}}

\providecommand{\palpha}{p_{\alpha}}
\def\DD{\mathcal{D}^+_{\text{init}}}
\def\DDr{{_{r}}\mathcal{D}^{+}_{\text{init}}}
\def\DDl{{_{l}}\mathcal{D}^{+}_{\text{init}}}

\def\avint{\mathop{\,\rlap{-}\!\!\int}\nolimits}

\def\g{\mathbf{g}}
\def\gma{\g_{M,a}}
\def\gme{\g_{M,e}}

\theoremstyle{plain}
\newtheorem{thm}{Theorem}[section]
\newtheorem{cor}[thm]{Corollary}
\newtheorem{lemma}[thm]{Lemma}
\newtheorem{prop}[thm]{Proposition}

\theoremstyle{definition}

\newtheorem{remark}[thm]{Remark}

\title{Precise late-time asymptotics of scalar field in the interior of a subextreme Kerr black hole and its application in Strong Cosmic Censorship conjecture}

\author[S. Ma and L.Zhang]{Siyuan Ma$^\dagger$ and Lin Zhang$^\star$}
\email{$^\dagger$siyuan.ma@aei.mpg.de, $^\star$lzhang\_math@cqu.edu.cn}
\address{$^\dagger$Laboratoire Jacques-Louis Lions,
Sorbonne Universit\'{e},
4 place Jussieu 75005 Paris, France\\
and\\
Max Planck Institute for Gravitational Physics, Am M\"{u}hlenberg 1, 14476 Potsdam,
Germany\\
and\\
Academy of Mathematics and Systems Science, The Chinese Academy of Sciences, Beijing 100190, China\\
$^\star$College of Mathematics and Statistics, Chongqing University, Chongqing 401331, China.}

\begin{document}

\allowdisplaybreaks

\begin{abstract}

In this work, we compute the precise late-time asymptotics for the scalar field in the interior of a non-static subextreme Kerr black hole,
based on recent progress on deriving its precise asymptotics in the Kerr exterior region. 
This provides a new proof of the generic $H^1_{\text{loc}}$-inextendibility of the Kerr Cauchy horizon against scalar perturbations that is first shown by Luk--Sbierski \cite{LukSbi16InstabKerr}.  The analogous results in Reissner--Nordstr\"{o}m spacetimes are also discussed.

\end{abstract}

\maketitle

\tableofcontents


\section{Main statement and introduction}
\label{sect:intro}

A Kerr metric $\gma$  is a Lorentzian metric
that takes the following form in the Boyer--Lindquist coordinates $(t,r,\theta,\phi)$: 
\begin{align}\label{eq:KerrMetricBoyerLindquistCoord}
\gma ={} & - \di t^2  + \frac{q^2}{\Delta} \di r^2+ q^2 \di \theta^2+\frac{2Mr}{q^2}(a\sin^2\theta \di \phi -\di t)^2 +(r^2+a^2)\sin^2\theta \di \phi^2.
\end{align}
Here, $q=\sqrt{r^2+a^2\cos^2\theta}$,   $\Delta= r^2 -2Mr+a^2$, $M$ and $Ma$ are the mass and the angular momentum of the Kerr black hole. In this work, we consider the \textbf{non-static subextreme Kerr}, that is, we require $0<\abs{a}<M$. The two roots $r_+$ and $r_-$ of function $\Delta$, which satisfy $0<r_-<M <r_+$,  are the locations of the event horizon $\HH$ and the Cauchy horizon $\CHH$.

We fix a time orientation on the manifold $(\MM\doteq\mathbb{R}\times (r_-,r_+)\times \Sphere, \gma)$ by requiring $-\ptr_r$ to be future-directed, and we call this time-oriented manifold $(\MM, \gma)$ as the \textit{Kerr interior}.
The event horizon $\HH$ and the Cauchy horizon $\CHH$ can be properly attached as the boundaries to the manifold $\MM$ in Kruskal-type coordinates. $\HH$ is a union of two bifurcate parts $\Horizon$ and $\Horizon'$, and $\CHH$ is a union of two bifurcate parts $\CHorizon$ and $\CHorizon'$, and we denote their intersection spheres  by $\Sphere_{\HH}\doteq\Horizon\cap \Horizon'$ and $\Sphere_{\CHH}\doteq\CHorizon\cap \CHorizon'$. 
In fact, Kruskal-type coordinates can be defined such as to extend the manifold beyond the event horizon to include the exterior regions $\MM_{\text{ext}}$ and $\MM_{\text{ext}}'$ of the Kerr black hole, with the extended manifold still being a proper subset of the Kruskal maximal extension of the Kerr spacetime; see for instance \cite{Neil95GeometryofKerr}. See Figure \ref{fig:1} for the Penrose diagram of the Kerr interior and Kerr exterior regions.

\begin{figure}[htbp]
  \begin{minipage}[t]{0.5\linewidth}
  \begin{center}
\begin{tikzpicture}[scale=0.9]
  \fill[green!20] (0,4)--(-2,2)--(0,0)--(2,2)--cycle;
  \fill[red!30] (-2,2)--(-4,0)--(0,0)--cycle;
  \fill[red!30] (2,2)--(4,0)--(0,0)--cycle;
  \draw[] (0.05,0.05) -- (1.95,1.95);
  \draw[] (1.95,2.05) -- (0.05,3.95);
  \draw[] (-1.95,1.95) -- (-0.05,0.05);
  \draw[] (-1.95,2.05) -- (-0.05,3.95);
  \draw[dashed] (2.05,1.95) -- (4,0);
  \draw[dashed] (-4,0) -- (-2.05,1.95);
  \draw[] (2,2) circle (0.05);
  \draw[] (-2,2) circle (0.05);
  \draw[fill=black] (0,0) circle (0.05);
  \draw[fill=black] (0,4) circle (0.05);
  \draw[] (-2.5,1.5) arc(225:315:3.53 and 2.9);
  \draw[->] (3,2)--(1.4,0.95);
  \node at (0.2,-0.3) {$\Sphere_{\HH}$};
  \node at (0.2,4.3) {$\Sphere_{\CHH}$};
  \node at (0.85,1.3) {$\Horizon$};
  \node at (-0.85,1.3) {$\Horizon'$};
  \node at (1.25,3.3) {$\CHorizon$};
  \node at (-1.25,3.3) {$\CHorizon'$};
  \node at (2.25,2.2) {$i_+$};
  \node at (-2.25,2.2) {$i_+'$};
  \node at (0,2) {$\MM$};
  \node at (3.65,0.75) {$\mathcal{I}_+$};
  \node at (-3.65,0.75) {$\mathcal{I}_+'$};
  \node at (2,0.4) {$\MM_{\text{ext}}$};
  \node at (-2,0.4) {$\MM_{\text{ext}}'$};
  \node at (3.4,2) {$\Sigmainit$};
\end{tikzpicture}
\end{center}
\caption{Kerr interior and Kerr exterior regions}
\label{fig:1}
\end{minipage}%
\begin{minipage}[t]{0.5\linewidth}
  \begin{center}
\begin{tikzpicture}[scale=0.9]
  \fill[green!20] (0,4)--(-2.5,1.5) arc(225:315:3.53 and 2.9)--cycle;
  \fill[red!30] (-2.5,1.5) arc(225:258.5:3.53 and 2.9)--(-2,2)--cycle;
  \fill[red!30] (0.71,0.71) arc(281.5:315:3.53 and 2.9)--(2,2)--cycle;
  \draw[] (1.95,2.05) -- (0.05,3.95);
  \draw[] (-1.95,2.05) -- (-0.05,3.95);
  \draw[dashed] (2.05,1.95) -- (4,0);
  \draw[dashed] (-4,0) -- (-2.05,1.95);
  \draw[fill=white] (2,2) circle (0.05);
  \draw[fill=white] (-2,2) circle (0.05);
  \draw[fill=black] (0,0) circle (0.05);
  \draw[fill=black] (0,4) circle (0.05);
  \draw[] (0.05,0.05) -- (1.95,1.95);
  \draw[] (-1.95,1.95) -- (-0.05,0.05);
  \draw[] (-2.5,1.5) arc(225:315:3.53 and 2.9);
  \draw[very thick] (-2.5,1.5) arc(225:258.5:3.53 and 2.9);
  \draw[very thick] (0.71,0.71) arc(281.5:315:3.53 and 2.9);
  \node at (0.2,-0.3) {$\Sphere_{\HH}$};
  \node at (0.2,4.3) {$\Sphere_{\CHH}$};
  \node at (1.25,3.3) {$\CHorizon$};
  \node at (-1.25,3.3) {$\CHorizon'$};
  \node at (2.25,2.2) {$i_+$};
  \node at (-2.25,2.2) {$i_+'$};
  \node at (3.65,0.75) {$\mathcal{I}_+$};
  \node at (-3.65,0.75) {$\mathcal{I}_+'$};
  \node at (0,2) {$\DD$};
  \node at (1.8,0.6) {$\Sigma_{\text{ext}}$};
  \node at (-1.8,0.7) {$\Sigma'_{\text{ext}}$};
  \node at (0,0.464) {$\Sigma_{\text{int}}$};
\end{tikzpicture}
\end{center}
\caption{Two-ended initial value problem}
\label{fig:2}
\end{minipage}
\end{figure}

In this work, we derive the precise late-time asymptotics of a (real) scalar field $\psi$, solving the (rescaled) scalar wave equation 
\begin{align}
\label{eq:wave}
q^2 \Box_{\gma} \psi={}&q^2 \nabla^{\alpha}\nabla_{\alpha}\psi= 0,
\end{align}
in the Kerr interior $\MM$, with smooth, compactly supported initial data imposed on a two-ended hypersurface $\Sigmainit$. 
This initial hypersurface $\Sigmainit$ is the union of $\Sigma'_{\text{ext}}\doteq\{(\tb', \rb',\omega): \tb'=\text{const}, \rb'\geq r_+,\omega\in \Sphere\}$, $\Sigma_{\text{ext}}\doteq\{(\tb, \rb, \omega): \tb=\text{const}, \rb\geq r_+,\omega\in \Sphere\}$\footnote{Without loss of generality, we assume that on spheres $\Sigma'_{\text{ext}}\cap \Horizon'$ and $\Sigma_{\text{ext}}\cap \Horizon$, $u=2$ and $\ub=2$ are satisfied respectively, where $u$ and $\ub$ are defined as in Section \ref{sect:E--F}. In addition, $\tau$ and $\tau'$ are hyperboloidal time functions, and $\rho=\rho'=r$. The explicit way of defining the hyperboloidal time functions is not relevant in this work; instead, the only requirement for this hyperboloidal function is such that its level hypersurfaces are spacelike hypersurfaces that are transversal to the event horizon and asymptotic to the future null infinity. We refer to \cite{andersson2019stability,MZ21PLKerr} for an explicit construction of the hyperboloidal  time function in the Kerr black hole exterior regions.}
 and a spacelike hypersurface $\Sigma_{\text{int}}$ in $\MM$ connecting\footnote{The way how  this hypersurface $\Sigma_{\text{int}}$ connects these two spheres is not relevant, as long as this hypersurface is strictly spacelike.} the two spheres $\Sigma'_{\text{ext}}\cap \Horizon'$ and $\Sigma_{\text{ext}}\cap \Horizon$, with $(\tb', \rb',  \omega)$ and $(\tb,\rb,\omega)$ being hyperboloidal coordinates in $\MM_{\text{ext}}'$ and  $\MM_{\text{ext}}$ respectively.  This initial value problem is well-posed in the future maximal Cauchy development of $\Sigmainit$, which is denoted by $\mathcal{D}^+(\Sigmainit)$. We moreover denote $\DD\doteq \mathcal{D}^+(\Sigmainit)\cap\MM$, in which the precise asymptotics for scalar field will be shown. See Figure \ref{fig:2}.
 
The main motivation of this current work arises from its intimate connection to the \textbf{Strong Cosmic Censorship} (SCC) conjecture which concerns the fundamental question of determinism in the theory of General Relativity. The method developed here is expected to be employed to consider more complicated model problems, such as the linearized gravity, in addressing the SCC conjecture about the Kerr spacetimes. We guide the readers to Section \ref{sect:liter}  for more background on SCC conjecture.

To properly state our main result, let $u=r^*-t$ and $\ub=r^*+t$ be the retarded time and forward time in $\MM$ respectively,  where $r^*$ be the tortoise coordinate, see Section \ref{sect:uub} for the explicit definition. 
 Further, define $\Lxi\doteq \ptr_{t}$ and two principal null vectors\footnote{We remark that the null vectors  $e_3$, $e_4$  are  regular at both the right event horizon $\Horizon$ and the left Cauchy horizon $\CHorizon'$ and $e_{4}$ is degenerate at both the right Cauchy horizon $\CHorizon$ and the left event horizon $\Horizon'$. A similar argument applies to the vectors $e_3'$ and $e_4'$ by interchanging the primed and unprimed notations.}
\begin{align}
\label{def:e3e4:1}
e_3\doteq {}&\frac{1}{2}\bigg(\frac{r^2+a^2}{\Delta}\ptr_t +\frac{a}{\Delta}\ptr_{\phi}- \ptr_r\bigg), & e_4\doteq {}&\frac{1}{2}\bigg(\ptr_t +\frac{a}{r^2+a^2}\ptr_{\phi} +\frac{\Delta}{\R}\ptr_r\bigg).
\end{align}
Also, define
\begin{align}
	e_3'\doteq (-\mu) e_3,\qquad\qquad e_4'\doteq (-\mu)^{-1}e_4,
\end{align} 
where $\mu\doteq\frac{\Delta}{\R}$. 
Let $\bar{\nabla}$ be the set of tangential derivatives on $\Horizon$, and let $\bar{\nabla}'$ be the set of tangential derivatives on $\Horizon'$.

\begin{figure}[htbp]
  \begin{minipage}[t]{0.5\linewidth}
  \begin{center}
\begin{tikzpicture}
  \fill[blue!40] (-0.7,2.7) arc(90:19.8:2.83 and 1)--(0,4)--(-0.96,3.04) arc(218:218.5:23.1);
  \fill[red!60] (-0.7,2.7) arc(90:19.8:2.83 and 1)--(1.95,1.95)--(0.91,0.91) arc(225:218.5:23.1);
  \fill[white] (0,0)--(-1.1,1.1) arc(255:285:4.25 and 3.6)--cycle;
  \draw[] (0.05,0.05) -- (1.95,1.95);
  \draw[] (1.95,2.05) -- (0.05,3.95);
  \draw[] (-1.95,1.95) -- (-0.05,0.05);
  \draw[] (-1.95,2.05) -- (-0.05,3.95);
  \draw[] (-0.96,3.04) arc (218:225:23.1);
  \draw[] (-0.7,2.7) arc(90:19.8:2.83 and 1);
  \draw[] (2,2) circle (0.05);
  \draw[] (-2,2) circle (0.05);
  \draw[] (-1.1,1.1) arc(255:285:4.25 and 3.6);
  \draw[fill=black] (0,0) circle (0.05);
  \draw[fill=black] (0,4) circle (0.05);
  \node at (0.2,-0.3) {$\Sphere_{\HH}$};
  \node at (0.2,4.3) {$\Sphere_{\CHH}$};
  \node at (1.25,3.3) {$\CHorizon$};
  \node at (-1.25,3.3) {$\CHorizon'$};
  \node at (2.25,2.15) {$i_+$};
  \node at (-2.25,2.15) {$i_+'$};
  \node[rotate=-52] at (-0.17,1.83) {$\ub=1$};
  \node at (0,3.1) {$2r^*>\ub^{\gamma_0}$};
  \node at (1,1.85) {$2r^*<\ub^{\gamma_0}$};
  \draw[->] (0.5,0.3)--(1.5,1.3);
  \node at (1,0.5) {$\ub$};
  \node at (0,0.8) {$\Sigma_{\text{int}}$};
\end{tikzpicture}
\end{center}
\caption{Region $\DDr$ and its subregions}
\label{fig:3}
\end{minipage}%
\begin{minipage}[t]{0.5\linewidth}
  \begin{center}
\begin{tikzpicture}
  \fill[blue!40] (0.7,2.7) arc(90:160.2:2.83 and 1)--(0,4)--(0.96,3.04) arc(-38:-38.5:23.1);
  \fill[red!60] (0.7,2.7) arc(90:160.2:2.83 and 1)--(-1.95,1.95)--(-0.91,0.91) arc(-45:-38.5:23.1);
  \fill[white] (0,0)--(-1.1,1.1) arc(255:285:4.25 and 3.6)--cycle;
  \draw[] (0.05,0.05) -- (1.95,1.95);
  \draw[] (1.95,2.05) -- (0.05,3.95);
  \draw[] (-1.95,1.95) -- (-0.05,0.05);
  \draw[] (-1.95,2.05) -- (-0.05,3.95);
  \draw[] (0.96,3.04) arc (-38:-45:23.1);
  \draw[] (0.7,2.7) arc(90:160.2:2.83 and 1);
  \draw[] (2,2) circle (0.05);
  \draw[] (-2,2) circle (0.05);
  \draw[] (-1.1,1.1) arc(255:285:4.25 and 3.6);
  \draw[fill=black] (0,0) circle (0.05);
  \draw[fill=black] (0,4) circle (0.05);
  \node at (0.2,-0.3) {$\Sphere_{\HH}$};
  \node at (0.2,4.3) {$\Sphere_{\CHH}$};
  \node at (1.25,3.3) {$\CHorizon$};
  \node at (-1.25,3.3) {$\CHorizon'$};
  \node at (2.25,2.15) {$i_+$};
  \node at (-2.25,2.15) {$i_+'$};
  \node[rotate=45] at (0.15,1.85) {$u=1$};
  \node at (0,3.1) {$2r^*>u^{\gamma_0}$};
  \node at (-1,1.85) {$2r^*<u^{\gamma_0}$};
  \draw[->] (-0.5,0.3)--(-1.5,1.3);
  \node at (-1,0.5) {$u$};
  \node at (0,0.8) {$\Sigma_{\text{int}}$};
\end{tikzpicture}
\end{center}
\caption{Region $\DDl$ and its subregions}
\label{fig:4}
\end{minipage}
\end{figure}

\begin{thm}[Global precise late-time asymptotics for the scalar field and its derivatives in Kerr interior]\label{thm:main}
Let $0<\abs{a}<M$. Let $\psi$ be a solution to the scalar wave equation \eqref{eq:wave} arising from smooth, compactly supported initial data on $\Sigmainit$. 

Let $j_1\in \Naturals$, $s_1>0$ be suitably large\footnote{By examining the proof, we find that $j_1=4$ and $s_1=9$ are sufficient. The reason that we loss much regularity is partly due to the application of Sobolev-imbedding on spheres to obtain pointwise decay estimates from energy decay estimates.}. 
Assume that for $\psi_{\ell\geq 1}=\psi -\avint_{\Sphere}\psi$,\footnote{$\psi_{\ell\geq 1}$ and $\avint_{\Sphere}\psi$ are actually the $\ell\geq 1$ spherical harmonic modes and the $\ell=0$ spherically symmetric mode of $\psi$, respectively.} where $\avint_{\Sphere}\psi$ is the spherical mean of $\psi$, there exists a $\delta>0$ and a $D_0\geq 0$ such that  all multiindex $\reg$ with $\abs{\reg}\leq s_1$ and all $0\leq j\leq j_1$,
\begin{align}
\label{assump:eq:HH:highmodes}
|\Lxi^j\bar{\nabla}^\reg\ellmode{\psi}{\ell\geq 1}|\leq D_0\ub^{-4-j-\delta}, \quad \text{on } \Horizon\cap\DD,\\
\label{assump:eq:HH:highmodes:LHS}
|\Lxi^j(\bar{\nabla}')^\reg\ellmode{\psi}{\ell\geq 1}|\leq D_0 u^{-4-j-\delta}, \quad \text{on } \Horizon'\cap\DD.
\end{align}
Assume moreover that there are real constants $c_0$ and $c_0'$ and an $\eps>0$\footnote{Without loss of generality, we assume $\eps\leq \delta$.} such that for all $0\leq j\leq j_1$,
\begin{subequations}
\label{ass:thm:horizon:psi}
\begin{align}
\label{eq:asympp:rred}
|\Lxi^{j}\psi - c_0 \Lxi^j(\ub^{-3})|\leq D_0 \ub^{-3-j-\eps}, \quad \text{on } \Horizon\cap\DD,\\
\label{eq:asympp:rred:LHS}
|\Lxi^{j}\psi- c_0' \Lxi^j(u^{-3})|\leq D_0 u^{-3-j-\eps}, \quad \text{on } \Horizon'\cap\DD.
\end{align}
\end{subequations}
Then, in region $\DDr\doteq\DD\cap\{\ub\geq 1\}$,\footnote{Figures \ref{fig:3} and \ref{fig:4} have depicted the different regions stated in this theorem.}
 \begin{enumerate}[label=\arabic*)]
 \item
 \label{item:mainthm:1}
 for $\psi$ itself, there exists a smooth function $\Psi(u,\omega)$ and $\Upsilon_0(\omega)$,  $\omega$ being the spherical coordinates on $\Sp$, such that\footnote{We use the notation $A\lesssim B$ to denote the bound of the form $A\leq CB$, where $A$ and $B$ are nonnegative and $C$ is a positive constant depending only on the black hole parameters $a$ and $M$ and the constant $D_0$ appearing in \eqref{ass:thm:horizon:psi}.}
 \begin{subequations}
  \label{eq:main:R:psi}
 \begin{align}
 \label{eq:main:R:psi:I}
    |\psi-c_0\ub^{-3}|\lesssim {}& \ub^{-3-\epsilon}\ \ \text{in}\ \DDr\cap \{2r^*\leq \ub^{\gamma_0}\},\\
 \label{eq:main:R:psi:II}
    \biggl|\psi-\Psi(u,\omega)-\half c_0\bigg(1+\frac{r_+^2+a^2}{r^2+a^2}\bigg)\ub^{-3}\biggr|\lesssim {}&\ub^{-3-\epsilon} \ \ \text{in}\ \DDr\cap \{2r^*\geq \ub^{\gamma_0}\} ,
  \end{align}
  where $\gamma_0\in (0,1)$ is an arbitrary constant and
   \begin{align}
 \label{eq:main:R:psi:II:1}
 \biggl|\Psi(u,\omega) +\half c_0\bigg(1-\frac{r_+^2+a^2}{r_-^2+a^2}\bigg){u}^{-3}\biggr|\lesssim{}& \abs{u}^{-3-\epsilon}\ \ \text{as}\ u\to-\infty,\\
 \label{eq:main:R:psi:II:2}
    \biggl|\Psi(\gamma_\omega(u))-\Upsilon_0(\omega)-\half c_0'\bigg(1+\frac{r_+^2+a^2}{r_-^2+a^2}\bigg){u}^{-3}\biggr|\lesssim {}& u^{-3-\epsilon} \ \ \text{as}\ u\to +\infty,
  \end{align}
 \end{subequations}
where $\gamma_\omega(u)$ is the integral curve of $e_3'\big|_{\CHorizon}$ starting from the sphere $\{u=1\}\cap\CHorizon$;
 
 \item for $e_4\psi$, we have
 \begin{align}
    \label{eq:main:R:e4psi}
      \bigg|e_4\psi+\frac{3}{2}c_0\bigg(1+\frac{r_+^2+a^2}{r^2+a^2}\bigg)\ub^{-4}\bigg|\lesssim \ub^{-4-\epsilon} \ \ \text{in}\ \DDr ;
    \end{align}

\item for $e_3'\psi$, we have
\begin{subequations}
    \label{eq:main:R:e3psi}
\begin{align}
\label{eq:main:R:e3psi:I}
    \bigg|e_3'\psi-\frac{3}{2}c_0\bigg(1-\frac{r_+^2+a^2}{r^2+a^2}\bigg)\ub^{-4}\bigg|\lesssim {}&(r_+-r) \ub^{-4-\epsilon} &&\ \ \text{in}\ \DDr\cap \{2r^*\leq \ub^{\gamma_0}\},\\
\label{eq:main:R:e3psi:II}
    \biggl|e_3'\psi-e_3'\vert_{\CHorizon}(\Psi(u,\omega))\biggr|\lesssim {}& -\mu &&\ \ \text{in}\ \DDr\cap \{r^*\geq 0\},
  \end{align}
  where $\Psi(u,\omega)$ and $\gamma_0$ are as in Point \ref{item:mainthm:1} and it holds that
  \begin{align}
\label{eq:main:R:e3psi:II:1}
\bigg|e_3'\vert_{\CHorizon}(\Psi(u,\omega))-\frac{3}{2}c_0\bigg(1-\frac{r_+^2+a^2}{r_-^2+a^2}\bigg){u}^{-4}\bigg|&\lesssim \abs{u}^{-4-\epsilon} \ \ \text{as}\ u\to -\infty.\\
\bigg|e_3'\vert_{\CHorizon}(\Psi(u,\omega))+\frac{3}{2}c_0'\bigg(1+\frac{r_+^2+a^2}{r_-^2+a^2}\bigg){u}^{-4}\bigg|&\lesssim u^{-4-\epsilon} \ \ \text{as}\ u\to +\infty.
\end{align}
\end{subequations}
\end{enumerate}

Meanwhile, there exist  smooth functions $\Psi'(\ub, \omega)$ and $\Upsilon_0'(\omega)$ such that the above estimates are valid in $\DDl\doteq\DD\cap \{u\geq 1\}$ if we make the replacements $u\rightarrow \ub$, $\ub\rightarrow u$, $e_3'\rightarrow e_4$, $e_4 \rightarrow e_3' =-\mu e_3$, $\Psi(u,\omega)\rightarrow \Psi'(\ub,\omega)$, $\Upsilon_0(\omega)\rightarrow \Upsilon_0'(\omega)$, $\Psi(\gamma_\omega(u))\rightarrow \Psi'(\gamma_\omega(\ub))$, $\DDr\rightarrow \DDl$, $\{2r^*\leq \ub^{\gamma_0}\}\rightarrow \{2r^*\leq u^{\gamma_0}\}$, $\CHorizon\rightarrow \CHorizon'$,  respectively.

Further, each of these estimates remains valid by taking $\Lxi$ on all terms (except that $\Upsilon_0(\omega)$ is replaced by some function $\Upsilon_j(\omega)$ if taking $j$ times $\Lxi$ derivatives and that the term $-\mu$ on the right-hand side of \eqref{eq:main:R:e3psi:II} remains undifferentiated).
\end{thm}

A couple of remarks are ready to illustrate that Theorem \ref{thm:main} indeed contains the global precise late-time asymptotics for the scalar field, as well as its principal null derivatives, in $\DD$.

\begin{remark}[Global late-time precise asymptotics in $\DDl\cup\DDr$]
\label{rem:DDrcupDDl}
From the above theorem, we have $|(-\mu e_3)\psi+\frac{3}{2}c_0'(1+\frac{r_+^2+a^2}{r^2+a^2})u^{-4}|\lesssim u^{-4-\epsilon}$ in $\DDl$ which is an analog of the estimate \eqref{eq:main:R:e4psi} by making the replacements as stated in the theorem, hence this estimate provides the asymptotics of $(-\mu e_3)\psi$ in the region $\DDl\cap\DDr\cap\{r^*\geq 0\}$. This complements the estimates \eqref{eq:main:R:e3psi:II}--\eqref{eq:main:R:e3psi:II:1} and provides global precise asymptotics for $(-\mu e_3)\psi$ in $\DDr$. Therefore, this provides the global precise late-time asymptotics of $\psi$ and its principal null derivatives in $\DDr$, and similarly in $\DDl$.
\end{remark}

\begin{remark}[Global late-time precise asymptotics for $\psi$ and its principal null derivatives in $\DD$]
Theorem \ref{thm:main} actually provides global precise late-time asymptotics for $\psi$ and its principal null derivatives in $\DD$, since the asymptotics in $\DDl\cup\DDr$ is proved as demonstrated in Remark \ref{rem:DDrcupDDl} and the remaining region $\DD\cap\{\ub\leq 1\}\cap\{u\leq 1\}$ is a compact region with both $u$ and $\ub$ uniformly bounded from above and below.
\end{remark}

We have chosen to state the asymptotics for the null derivatives $e_4$ and $(-\mu e_3)$ in Theorem \ref{thm:main}. Nevertheless, Theorem \ref{thm:main} has already included the late-time asymptotics for the regular, nondegenerate derivative $e_3\psi$ near the right event horizon. This is explained in the following remark. 

\begin{remark}[Precise late-time asymptotics for nondegenerate transverse derivative near event horizon]
The above estimate \eqref{eq:main:R:e3psi:I} in fact includes the precise late-time asymptotics for the regular derivative $e_3\psi$ near the right event horizon $\Horizon$. Indeed, one finds $|e_3\psi+\frac{3}{2}c_0\frac{r+r_+}{r-r_-}\ub^{-4}| \lesssim \ub^{-4-\epsilon}$ in $\DDr\cap\{r_0\leq r\leq r_+\}$ for any given $r_0\in (r_-, r_+)$. Similarly, it holds that $|(-\mu)^{-1}e_4 \psi+\frac{3}{2}c_0\frac{r+r_+}{r-r_-}u^{-4}| \lesssim u^{-4-\epsilon}$ in $\DDl\cap\{r_0\leq r\leq r_+\}$.
\end{remark}

With minor modification in the proof of Theorem \ref{thm:main}, we obtain an analogous result in the charged Reissner--Nordstr\"{o}m  spacetimes. We state it here but omit the proof.

\begin{thm}[Global precise late-time asymptotics for scalar field in the interior of a Reissner--Nordstr\"{o}m spacetime]
In the interior $(\MM, \gme)$ of a charged subextreme R--N spacetime, where $M$ and $e$ are the mass and charge of the black hole and satisfy $0<\abs{e}<M$, the statements in Theorem \ref{thm:main} are valid for the scalar field if we set $a=0$ and $\Delta=r^2-2Mr+e^2$ in both the formula \eqref{def:e3e4:1} of $e_3$ and $e_4$ and the estimates \eqref{eq:main:R:psi}--\eqref{eq:main:R:e3psi}, still let $r_{\pm}$ be the roots of function $\Delta$, and redefine $\mu=\frac{\Delta}{r^2}=\frac{r^2-2Mr-e^2}{r^2}$. 
\end{thm}

The assumed estimates \eqref{ass:thm:horizon:psi} for $\psi$ on $\Horizon\cap\DD$ and $\Horizon'\cap\DD$ in Theorem \ref{thm:main} are recently shown to hold true, and the constants $c_0$ and $c_0'$ can be computed from the initial data on $\Sigma_{\text{ext}}$ and $\Sigma'_{\text{ext}}$, respectively.

\begin{thm}[Precise asymptotics for the scalar field on the Kerr event horizon \cite{hintz2022sharp,angelopoulos2021late,MZ21PLKerr}]
\label{thm:PL}
For scalar field $\psi$ arising from smooth, compactly supported initial data on a two-ended hypersurface $\Sigmainit$, 
the assumed estimates \eqref{ass:thm:horizon:psi} on $\Horizon\cap\DD$ and $\Horizon'\cap\DD$ are satisfied for any $0<\delta\leq 1$ and
\begin{align}
\label{eq:c0:value}
c_0=&-\frac{2M}{\pi} \bigg((r_+^2+a^2)\int_{\Sigma_{\text{ext}}\cap \Horizon} \psi \di \omega -\int_{\Sigma_{\text{ext}}}
\langle\nabla\tau, \nabla\psi\rangle_{\gma} q^2\di\rho\di \omega\bigg),\\
\label{eq:c0':value}
c_0'=&-\frac{2M}{\pi} \bigg((r_+^2+a^2)\int_{\Sigma'_{\text{ext}}\cap \Horizon'} \psi \di \omega -\int_{\Sigma'_{\text{ext}}}\langle\nabla\tau', \nabla\psi\rangle_{\gma} q^2\di \rho'\di\omega\bigg).
\end{align}
Further, for generic such initial data, the constants $c_0$ and $c_0'$ are non-zero.
\end{thm}

\begin{remark}[Generic inextendibility for scalar field in Kerr interior]
Theorems \ref{thm:main} and \ref{thm:PL} together imply that the regular derivative $(-\mu)^{-1} e_4 \psi$ generically blows up towards the right Cauchy horizon $\CHorizon$ and the nondegenerate energy of $\psi$ on hypersurface $\Cu_{w}\cap \{\ub\geq \ub_0\}$\footnote{$\Cu_{w}$ is a constant-$w$ spacelike hypersurface with $w=u-r+r_-$.}, which in particular bounds $\int_{\Cu_{w}\cap\{\ub\geq \ub_0\}} \abs{\mu}^{-1}\abs{e_4\psi}^2 \di \ub$, generically goes to $+\infty$ for any $\ub_0$. This proves the SCC conjecture for the model of scalar field in the class of $H^1_{\text{loc}}$, an argument that is first demonstrated by Luk--Sbierski \cite{LukSbi16InstabKerr}.\footnote{Theorems \ref{thm:main} and \ref{thm:PL} together also show that the scalar field is $C^0$-extendible to the Cauchy horizon; this argument has been shown by Hintz \cite{hintz2017boundedness} and Franzen \cite{franzen2016boundedness,franzen2020boundedness} independently.} The same argument applies to the Reissner--Nordstr\"om spacetimes. Further, these precise asymptotics can be used to examine the validity of SCC conjecture for scalar field in a given weak regularity space. 
\end{remark}

Let us briefly discuss the above two theorems in the following Figure \ref{fig:5}. The statements in Theorem \ref{thm:PL} are actually to prove the so-called ``Price's law (PL)'' (see more in Section \ref{sect:liter}) in the exterior of a Kerr black hole and correspond to the green-colored region. The statements in Theorem \ref{thm:main} contain the precise asymptotics of the scalar field in the red-colored and blue-colored regions, in which the red-shift and blue-shift effects are present respectively. In particular, the blue-colored region is where we show the $H^1_{\text{loc}}$-inextendibility.

\begin{figure}[htbp]
  \begin{center}
\begin{tikzpicture}
  \fill[green!30] (2,2)--(1.3,1.3) arc (225:314.6:1)--cycle;
\fill[green!30] (-2,2)--(-2.75,1.25) arc (225:315:1.06)--cycle;
  \fill[red!60] (-2,2)--(-1.3,1.3) arc (-45:45:1)--cycle;
  \fill[red!60] (2,2)--(1.4,2.6) arc(135:215:0.8 and 1)--cycle;
  \fill[blue!40] (-2,2) cos (0.3,3.7)--(0,4)--cycle;
  \fill[blue!40] (2,2) cos (-0.3,3.7)--(0,4)--cycle;
  \draw[] (0.05,0.05) -- (1.95,1.95);
  \draw[] (1.95,2.05) -- (0.05,3.95);
  \draw[] (-1.95,1.95) -- (-0.05,0.05);
  \draw[] (-1.95,2.05) -- (-0.05,3.95);
  \draw[dashed] (2.05,1.95) -- (4,0);
  \draw[dashed] (-4,0) -- (-2.05,1.95);
  \draw[] (2,2) circle (0.05);
  \draw[] (-2,2) circle (0.05);
  \draw[fill=black] (0,0) circle (0.05);
  \draw[fill=black] (0,4) circle (0.05);
  \node at (0.2,-0.3) {$\Sphere_{\HH}$};
  \node at (0.2,4.3) {$\Sphere_{\CHH}$};
  \node at (0.75,1.2) {$\Horizon$};
  \node at (-0.75,1.2) {$\Horizon'$};
  \node at (1.25,3.3) {$\CHorizon$};
  \node at (-1.25,3.3) {$\CHorizon'$};
  \node at (2.25,2.2) {$i_+$};
  \node at (-2.25,2.2) {$i_+'$};
  \node at (0,2) {$\MM$};
  \node at (3.65,0.75) {$\mathcal{I}_+$};
  \node at (-3.65,0.75) {$\mathcal{I}_+'$};
  \node at (2,0.3) {$\MM_{\text{ext}}$};
  \node at (-2,0.3) {$\MM_{\text{ext}}'$};
  \node at (2,1.45) {$\text{PL}$};
  \node at (-2,1.45) {$\text{PL}$};
  \node[rotate=45] at (-0.85,2.85) {$\text{SCC}$};
  \node[rotate=-45] at (0.85,2.85) {$\text{SCC}$};
\end{tikzpicture}
\end{center}
\caption{Different asymptotic regions in the exterior and interior of a Kerr black hole}
\label{fig:5}
\end{figure}
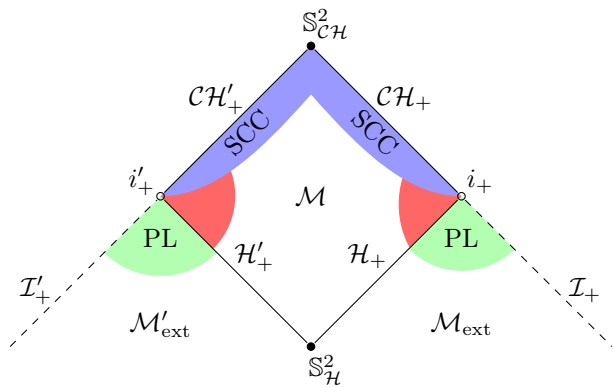

In the remainder of this section, we explain the Strong Cosmic Censorship conjecture and review the relevant works in the literature.

\subsection{Literature on the Strong Cosmic Censorship conjecture}
\label{sect:liter} 

It is well-known that the Kerr metric and the Reissner--Nordstr\"{o}m (R--N) metric  can be extended smoothly beyond the Cauchy horizon in non-unique ways. The SCC conjecture, proposed by Penrose \cite{penrose1974gravitational}, concerns the global uniqueness  for the initial value problem of the Einstein's equation and is currently formulated as:
\\

\fbox{%
  \parbox{0.9\textwidth}{%
\textbf{SCC Conjecture:}
\textit{For generic vacuum, asymptotically flat initial data, the maximal Cauchy development is inextendible as a suitably regular Lorenztian manifold.}}
}
\\

This conjecture suggests that the above scenario in exact R--N and Kerr spacetimes that it admits non-uniques smooth extensions beyond the Cauchy horizon is not generic and, in particular, is unstable in a suitable regularity sense against small initial data perturbations. 

For relevant physics literature towards the SCC conjecture, see for instance \cite{MaNamara78InstabInnerHori,Hiscock81InteriorCharge,chandrasekhar1982crossing,PoissonIsrael89InstabMassInflation,Ori91InnerCharged,ori1998evolution,Cardoso18PRL}. Here, we provide a brief overview of the mathematical results on this conjecture in the literature. 

The first complete result on this problem is for the Einstein--Maxwell--(real) scalar field system under spherical symmetry. Dafermos \cite{dafermos2003stability} and Dafermos--Rodnianski \cite{dafermos2005proof} showed the existence of the Cauchy horizon and proved the $C^0$-extendibility across the Cauchy horizon. Afterwards, Luk--Oh \cite{lukoh19SCCI,lukoh19SCCII} proved generic $C^2$-inextendibility for this model based on deriving  the precise asymptotics for the scalar field in the black hole exterior; this result is recently improved by Sbierski \cite{Sbierski21Holonomy}  to $C^{0,1}_{\text{loc}}$-inextendibility. See also the works \cite{Moortel18stabilityRN,Moortel21mass} by Van de Moortel for different nonlinear model problems under spherical symmetry.

The above results are however restricted to the spherical symmetry assumption, which itself does not generically hold. For the full Einstein's equation without imposing any symmetry condition, the only available result is by Dafermos--Luk \cite{DafLuk17C0stability} who proved the $C^0$-stability of the Kerr Cauchy horizon under suitable assumptions on the dynamical metric settling down to Kerr on a hypersurface inside the black hole.

Alternatively, one can consider simplified models of the Einstein's equation and impose no symmetry condition, and the scalar wave equation is a simple ideal model equation to start with. 

The first result for the scalar field along this direction is by Franzen \cite{franzen2016boundedness} who proved the $C^0$-extendiblity across the R--N Cauchy horizon. This result is further extended to Kerr by Hintz \cite{hintz2017boundedness} and Franzen \cite{franzen2020boundedness}.
On the other hand, it is expected from the SCC conjecture that the field shall be generically intextendible in a suitable regularity class. Sbierski \cite{sbierski2015characterisation}  used the Gaussian beam approximation to construct a sequence of solutions to the scalar wave equation whose $H^1$ energy near the Kerr (or R--N) Cauchy horizon tends to infinity. See also \cite{Dafermos17Blueshi}. The first generic $H^1_{\text{loc}}$-inextendibility in any neighborhood near the  Cauchy horizon is proved by Luk--Oh \cite{Luk2017LinInstabRN} in  R--N spacetimes. Afterwards, Luk--Sbierski \cite{LukSbi16InstabKerr} imposed suitable upper and lower bounds of decay for the horizon flux of the scalar field on the event horizon of  the Kerr spacetimes and showed the $H^1_{\text{loc}}$-inextendibility. In these two $H^1_{\text{loc}}$-inextendibility results, they both rely on conservation laws associated to Killing vector fields to derive the lower bound for the energy flux on a hypersurface transversally asymptotic to the Cauchy horizon. Recently, Luk--Oh--Shlapentokh-Rothman \cite{Luk22ScatteringCauchyHor} reproved the $H^1_{\text{loc}}$-inextendibility result of \cite{Luk2017LinInstabRN} in R--N spacetimes by analyzing the scattering map near $0$ time-frequency, and Sbierski \cite{Sbierski22InstabKerrCH} adapted this approach to show the instability of the Kerr Cauchy horizon for the linearized gravity by assuming suitable integrated upper and lower bounds of decay on the event horizon\footnote{These assumed bounds of decay on event horizon are satisfied in view of our work \cite{MZ21PLKerr}.}.

This generic $H^1_{\text{loc}}$-inextendibility in fact represents the so-called ``weak null singularity''  proposed by Christodoulou \cite{christodoulou2002global}, which prevents the metric to be interpreted as even a weak solution to the Einstein's equation. Luk \cite{Luk2018weaknull} constructed the weak null singularities in vacuum and showed that they locally propagate.

As discussed already, the above results on $H^1_{\text{loc}}$-inextendibility for the scalar field on Kerr and R--N are based on suitable lower bounds of decay on the event horizon. These lower bounds of decay rates correspond to the so-called ``Price's law'' \cite{Price1972SchwScalar,Price1972SchwIntegerSpin,bo99,price2004late,gleiser2008late} which predicts that the sharp decay rate, as both an upper and a lower bounds, is $\ub^{-3}$ in any finite radius region in the Kerr exterior and on the event horizon. 

This upper bound of decay for the scalar field is verified by Tataru \cite{tataru2013local} and Donninger--Schlag--Soffer \cite{donninger2011proof}. The corresponding sharp upper bound of decay rates  on  dynamic backgrounds is obtained by Metcalfe--Tataru--Tohaneanu \cite{metcalfe2012price}, and recently by Looi \cite{looi2022improved}. Recently, Hintz \cite{hintz2022sharp} and Angelopoulos--Aretakis--Gajic \cite{angelopoulos2021late,angelopoulos2021price} derived the precise late-time asymptotics and confirmed the Price's law for scalar field outside the Kerr and R--N black holes. We \cite{Ma20almost,MZ2022Dirac,MaZhang21PriceSchw,MZ21PLKerr} built on the the first author's works \cite{Ma2017Maxwell,Ma17spin2Kerr} and derived the precise late-time asymptotics, thus proving the corresponding Price's law, for the solutions to the Teukolsky equation \cite{Teu1972PRLseparability}, which governs the dynamics of the scalar field, the massless Dirac field, the electromagnetic field and the linearized gravity, in the  exterior and on the event horizon of Schwarzschild and Kerr black holes.

We end this subsection with a brief summary of the works on the nonlinear stability of the Kerr exterior. On one hand, the Price's law is intimated connected to the Kerr exterior stability problem; on the other hand, a complete proof of the Kerr exterior stability serves as an indispensable precursor for a full resolution of the SCC conjecture about Kerr. The linear stability of Schwarzschild and subextreme  R--N spacetimes has been shown in \cite{dafermos2019linear,Giorgi2019linearRNfullcharge}, and linear stability of Kerr spacetimes is proved in \cite{andersson2019stability,hafner2021linear,andersson2021nonlinear,Andersson22ModeKerrLargea}. 
For  nonlinear stability results, we refer to \cite{klainermanszeftel2020global,dafermos2021non} for Schwarzschild and \cite{HintzKds2018,fang2021nonlinear,fang2022linear}  for  slowly rotating Kerr-de Sitter. The nonlinear stability of slowly rotating Kerr is recently proved  in a series of works by Giorgi--Klainerman--Szeftel \cite{klainerman2019constructions,giorgi2020general,klainerman2021kerr,GKS22WaveNonlinearKerr} and \cite{Shen22GCMHyper}.

\subsection{Sketch of the proof}

To obtain precise asymptotics in the interior region of Kerr spacetime, we decompose the scalar field $\psi$ into its lower mode part $\psi_{\ell=0}$ and higher mode part $\psi_{\ell\geq 1}$. By assuming suitable late-time behavior for each part on the event horizon, we can prove that the higher mode part decays faster than the lower mode part in the black hole interior region. Consequently, we can determine the leading order term of $\psi$ solely from the asymptotics of the lower mode part $\psi_{\ell=0}$, which can in turn be obtained through an ODE argument. 
This technique was  employed in \cite{angelopoulos2021price,angelopoulos2021late} to derive the precise late-time asymptotics of the scalar field in the exterior region of the R--N and Kerr spacetimes and by us \cite{Ma20almost,MaZhang21PriceSchw, MZ21PLKerr} to study the precise late-time behaviors of the solutions to the Teukolsky equation \cite{Teu1972PRLseparability} in Schwarzschild and Kerr spacetimes. 
The current paper  builds on the detailed information (sharp decay for modes) of the scalar field on the event horizon  established by previous works, and extends the idea beyond the event horizon, thus completing the precise asymptotics globally in the Kerr spacetimes.

To clarify the method, let us briefly outline the procedure in computing the leading order term for only $T^j\psi_{\ell=0}$ and $\pv T^j\psi_{\ell=0}$ near the right Cauchy horizon $\CHorizon$. This can be achieved by integrating either $\pu(T^j\psi_{\ell=0})$ or the wave equation for $\psi_{\ell=0}$
\begin{align}\label{intro:00000}
	\hspace{4ex}&\hspace{-4ex}
	\pu\biggl((\R) \pv\psiz-\half(\R)\Lxi\psiz\biggr)
	=-\half(\R)\pu\Lxi\psiz-\frac{1}{4}a^2\mu \PJ_{\ell=0} (\sin^2\theta \Lxi^2\psi)
\end{align}
along $\ub=\text{constant}$ starting from the event horizon $\Horizon$ and using the energy decay estimates for $\int_{\ub=\text{constant}}(-\mu)^{-1}|\pv T^j\psi_{\ell=0}|^2\di u$. The proof is proceeded in three disjoint subregions $\I$, $\II_\Gamma$ and $\II\backslash\II_\Gamma$ as depicted in Figure \ref{fig:7}.

The first subregion $\I$ is the red-shift region, i.e., ${r\in (\rblue, r_+)}$. We can prove that
\begin{align}
	\int_{\{\ub=\ub_1\}\cap \I}(-\mu)^{-1}|\pv T^j\psi_{\ell=0}|^2\di u\lesssim \ub _1^{-8-2j},
\end{align}
which follows from standard red-shift estimates and the assumption of the asymptotics on the event horizon. Similar estimates have been derived in \cite{Luk2017LinInstabRN,LukSbi16InstabKerr}. The crucial point to keep in mind is that the right-hand side of \eqref{intro:00000} has faster decay because of the presence of the $T$ derivative, and hence can be regarded as error terms with faster decay.

The second subregion $\II_\Gamma$ is contained within the blue-shift region ${r\in (r_-,\rblue)}$ and satisfies $r^*(u,\ub)\leq \frac{1}{2}\ub^\gamma$ for a very small $\gamma>0$. We can prove that
\begin{align}
	\int_{\{\ub=\ub_1\}\cap \II_\Gamma}\lmu^{-\frac{1}{2}}|\pv T^j\psi_{\ell=0}|^2\di u\lesssim \ub _1^{-7-2j},
\end{align}
which follows from the blue-shift effect and the fact that $r^*(u,\ub)\lesssim \ub^\gamma$ for $(u,\ub)\in \II_\Gamma$.

To derive the asymptotics of $\pv\psi_{\ell=0}$ in the last subregion $\II\backslash\II_\Gamma$, we will integrate a modified version of \eqref{intro:00000}, which is given by
\begin{align}\label{intro:111111}
	\pu\bigl((\R)\pv\psiz\bigl)
	=\half \mu r\Lxi\psiz-\frac{1}{4}a^2\mu\PJ_{\ell=0} (\sin^2\theta \Lxi^2\psi),
\end{align}
and use the following energy bound:
\begin{align}
	\int_{\{\ub=\ub_1\}\cap \II\backslash\II_\Gamma} \lmu^{-\frac{5}{2}}|\Lxi^j\psiz|^2(u,\ub_1)\di u \lesssim {}&1
\end{align}
 and the fact  that $-\mu\lesssim e^{-c\ub^\gamma} $ in $\II\backslash\IIf$ to find that the right-hand side of \eqref{intro:111111} decays very rapidly.  Unlike in \cite{Luk2017LinInstabRN,LukSbi16InstabKerr} where the key part is to derive a quanlitative lower bound of $\int_{u=constant}\ub^p |\pv\psi|^2\di \ub$ using the blue-shift estimates, the benefit of using the ODE argument in our work is that we only need some suitable upper bound estimates in blue-shift region and exploit the exponential decay of $-\mu$ coefficient in region $\II\backslash\II_\Gamma$ to deduce the precise late-time asymptotics of $\pv \psiz$ (and thus $\pv \psi$).
Lastly, we note that $\psi_{\ell=0}$ can be continuously extended to the Cauchy horizon by simply using the estimates \eqref{eq:main:R:e4psi} for $\pv\psi_{\ell=0}$,  and  the precise late-time asymptotics for $\psi_{\ell=0}$ (for instance, \eqref{eq:main:R:psi:II} for  $\psi_{\ell=0}$) follow by integrating \eqref{eq:main:R:e4psi} along $u=constant$ starting from the Cauchy horizon. 

\subsection{Overview of the paper}

In Section \ref{sect:prel}, we discuss the geometry of Kerr spacetimes and introduce some preliminaries. We then prove the energy decay estimates and derive the precise asymptotics for the scalar field in a region closer to the event horizon in Section \ref{sect:stability} and  in the other region closer to the Cauchy horizon in Sections \ref{sect:EnerDec:II} and \ref{sect:precise:II}. In the last Section \ref{sect:pf:mainthm}, we combine these asymptotic estimates in different regions to complete the proof of our main Theorem \ref{thm:main}.

\section{Geometry of Kerr spacetimes and preliminaries}
\label{sect:prel}

In this section, we discuss various geometric aspects of the subextreme Kerr spacetimes, present the explicit form of the scalar wave equation, and introduce the mode decomposition.

\subsection{Scalar functions in the Kerr interior}
\label{sect:uub}

Recall that
$\mu=\frac{\Delta}{r^2+a^2}$. 
Note that function $\mu<0$ in $\MM$, $\mu=0$ on the horizons, and $\ptr_r\mu=\frac{2(r-M-\mu r)}{\R}=\frac{2M(r^2 -a^2)}{(\R)^2}$. 

Define a tortoise coordinate $r^*=r^*(r)$ by 
\begin{align}
\di r^*=\mu^{-1}\di r,\qquad \quad r^*(M)=0.
\end{align}
Note that function $r^*\to \mp\infty$ as $r\to r_{\pm}$. 

Let $\kappa_{+}=\frac{r_{+} - r_{-}}{4Mr_{+}}$ and
$\kappa_{-}=\frac{r_{-} - r_{+}}{4Mr_{-}}$ be the values of the surface gravity on the event and Cauchy horizons respectively. 
The asymptotics of $r^*$ towards the event and Cauchy horizons are given by
\begin{align}
\lim_{r\to r_+} \frac{r^*}{\frac{1}{2\kappa_+} \ln (r_+-r)}=1, \qquad \lim_{r\to r_-} \frac{r^*}{\frac{1}{2\kappa_-} \ln (r-r_-)}=1.
\end{align}

Define the forward time $\ub$ and retarded time $u$ functions by
\begin{align}
\ub\doteq{}& r^* + t,  \qquad \quad u\doteq r^*- t.
\end{align}
They satisfy $u+\ub=2r^*$.

In the Kerr interior, the constant-$r$ hypersurfaces are spacelike hypersurfaces since $\gma(\di r,\di r)=\frac{\Delta}{q^2}<0$, while the constant-$u$ and -$\ub$ hypersurfaces are timelike away from the poles and null at the poles since
$\gma(\di u, \di u) = \gma(\di \ub, \di \ub)=\frac{a^2\sin^2\theta}{q^2}\geq 0$. 
Instead, we define functions 
\begin{align}
w\doteq u-r+r_-, \qquad \quad \underline{w}\doteq \ub-r+r_+, 
\end{align}
and the constant-$w$ and -$\underline{w}$ hypersurfaces are spacelike with
\begin{align}
\gma(\di w, \di w)=\gma(\di \underline{w}, \di \underline{w})= -\frac{r^2+2Mr +a^2\cos^2\theta}{q^2} <-1.
\end{align}
The constant-$u$, -$\ub$, -$w$, and $-\underline{w}$ hypersurfaces are depicted in Figure \ref{fig:6}.

\begin{figure}[htbp]
  \begin{center}
\begin{tikzpicture}[scale=1.2]
  \draw[] (0.05,0.05) -- (1.95,1.95);
  \draw[] (1.95,2.05) -- (0.05,3.95);
  \draw[] (-1.95,1.95) -- (-0.05,0.05);
  \draw[] (-1.95,2.05) -- (-0.05,3.95);
  \draw[] (-1.06,2.94) arc(218:225:23.1);
  \draw[thick] (-1.3,2.7) arc(225:232:23.1);
  \draw[] (1.06,2.94) arc(-38:-45:23.1);
  \draw[thick] (1.3,2.7) arc(-45:-52:23.1);
  \draw[] (2,2) circle (0.05);
  \draw[] (-2,2) circle (0.05);
  \draw[fill=black] (0,0) circle (0.05);
  \draw[fill=black] (0,4) circle (0.05);
  \node at (0.2,-0.3) {$\Sphere_{\HH}$};
  \node at (0.2,4.3) {$\Sphere_{\CHH}$};
  \node[rotate=-45] at (1.1,3.2) {$\ub=+\infty$};
  \node[rotate=45] at (-1.1,3.2) {$u=+\infty$};
  \node[rotate=45] at (1.2,0.9) {$u=-\infty$};
  \node[rotate=-45] at (-1.2,0.9) {$\ub=-\infty$};
  \node at (2.2,2.2) {$i_+$};
  \node at (-2.2,2.2) {$i_+'$};
  \node[rotate=-51] at (-0.5,2.45) {$\ub= c$};
  \node[rotate=-42] at (-0.9,2.05) {$\tub= c$};
  \node[rotate=51] at (0.45,2.5) {$u= c'$};
  \node[rotate=42] at (0.9,2.05) {$\tu= c'$};
  \draw[->] (0,3.2)--(0.3,3.5);
  \draw[->] (0,3.2)--(-0.3,3.5);
  \node at (0.35,3.3) {$e_4$};
  \node at (-0.35,3.3) {$e_3$};
\end{tikzpicture}
\end{center}
\caption{$\ub=c$ and $u=c'$ are timelike hypersurfaces; $\tub=c$ and $\tu=c'$ are spacelike hypersurfaces}
\label{fig:6}
\end{figure}

\subsection{Operators and derivatives}

Recall that $\Lxi=\ptr_{t}$. Let us further define
\begin{align}
\Leta\doteq \partial_{\phi}, \qquad \quad \Th\doteq \ptr_{\th}.
\end{align}
When viewed as vector fields, $\Lxi$ and $\Leta$ are Killing vector fields of the Kerr spacetime. 

 Denote the Carter operator by 
 \begin{align}
 \Carter\doteq \Delta_{\Sphere}+a^2\sin^2\th\Lxi^2,
 \end{align}
   where $\Delta_{\Sphere}$ is the spherical Laplacian on unit $2$-sphere. This is a second-order symmetry operator for the scalar wave equation \eqref{eq:wave} in the sense that it commutes with equation \eqref{eq:wave}.

Recall formula \eqref{def:e3e4:1} of $e_3$ and $e_4$. Note that $e_3$ and $e_4$ are future-directed null vectors satisfying 
\begin{subequations}
\begin{align}
&\mu e_3 +e_4=\Lxi +\frac{a}{\R}\Leta,\\
&\gma(e_3, e_3)=\gma(e_4, e_4)=0, \quad \gma(e_3, e_4)=-\frac{q^2}{2(\R)},
\end{align}
\end{subequations}
and they are aligned with the principal null directions on the Kerr background. 

Further, it is convenient to define
\begin{align}
e_1\doteq \Th, \qquad e_2 \doteq \frac{1}{\sin\th}\Phi + a\sin\th \Lxi.
\end{align}
The frame $(\tilde{e}_3, \tilde{e}_4, \tilde{e}_1, \tilde{e}_2)=(2e_3, \frac{2(\R)}{q^2}e_4, q^{-1} e_1, q^{-1}e_2)$ constitutes a principal null frame of the Kerr background which satisfies $\gma(\tilde{e}_3, \tilde{e}_4)=-2$, $\gma(\tilde{e}_1, \tilde{e}_1)=\gma(\tilde{e}_2, \tilde{e}_2)=1$ and other products being zero.

\subsection{Eddington--Finkelstein coordinates}
\label{sect:E--F}

Apart from extending the Kerr metric beyond the Kerr interior $\MM$ using the Kruskal-type coordinates, the Kerr metric $\gma$ can be extended beyond $\MM$ in a different manner by using the Eddington--Finkelstein coordinates.

Define $\rmod=\rmod(r)$ by $\frac{\di\rmod}{\di r}=\frac{a}{\Delta}$ and $\rmod(M)=0$, and define
\begin{align}
\phiin\doteq {}&\phi +\rmod  \text{\quad mod\,\,} 2\pi, & \phiout\doteq{}&\phi -\rmod  \text{\quad mod\,\,} 2\pi.
\end{align}

We can extend $\MM$ to $\MM\cup (\Horizon\backslash \Sphere_{\HH})\cup (\CHorizon'\backslash \Sphere_{\CHH})$ using the ingoing Eddington--Finkelstein (E--F) coordinates $(\ubin=\ub,\rin=r, \thin=\th, \phiin)$, and  their partial derivatives in these ingoing E--F coordinates are related to the ones in B--L coordinates by 
\begin{align}
\ptr_{\ubin}=\ptr_{t}=\Lxi, \quad \ptr_{\rin}= \ptr_r -\frac{\R}{\Delta}\ptr_{t} - \frac{a}{\Delta} \ptr_{\phi}=-e_3, \quad \ptr_{\thin}=\ptr_{\th}=\Th, \quad \ptr_{\phiin}=\ptr_{\phi}=\Phi.
\end{align}

Analogously, the Kerr interior manifold $\MM$ can be extended to $\MM\cup (\Horizon'\backslash \Sphere_{\HH})\cup (\CHorizon\backslash \Sphere_{\CHH})$ using the outgoing E--F coordinates $(\uout=u,\rout=r, \thout=\th, \phiout)$, and one finds
\begin{align}
\ptr_{\uout}=\Lxi, \quad \ptr_{\rout}=\mu^{-1}e_4, \quad \ptr_{\thout}=\Th, \quad \ptr_{\phiout}=\Phi.
\end{align}

\subsection{Double-null-like coordinate systems}
\label{subsect:almostnull}

We define in $\MM$ a coordinate system $(u,\ub, \th,\phiin)$, and in this coordinate system,
\begin{align}
\label{eq:deris:uub:nullframe:I}
\pu = -\mu e_3, \quad \pv=\Lxi - \mu e_3=e_4 -\frac{a}{\R}\Phi, \quad \ptr_{\theta}=\Th,\quad \ptr_{\phiin}=\Leta.
\end{align}
We also define in $\MM$ a coordinate system $(u,\ub,\th,\phiout)$, and in these coordinates,
\begin{align}
\label{eq:deris:uub:nullframe:II}
\pu =e_4-\Lxi=-\mu e_3 +\frac{a}{\R}\Leta, \quad \pv=e_4, \quad \ptr_{\theta}=\Th,\quad \ptr_{\phiout}=\Phi.
\end{align}

We call these two coordinate systems as ``double-null-like'' coordinate systems for the reason that they are  in fact double null coordinates in the interior of a Schwarzschild black hole.

We should emphasis that these two different coordinates system will \underline{\textbf{only}} be used in the regions $\I\doteq\DD\cap\{\rblue\leq r\leq r_+\}\cap\{\ub\geq 1\}$ and $\II\doteq\DD\cap\{r_-\leq r\leq \rblue\}\cap\{\ub\geq 1\}$ that are close to $\Horizon$ and $\CHorizon$ (as well as regions $\I'$ and $\II'$ close to  $\Horizon'$ and  $\CHorizon'$, as defined in the beginning of Sections \ref{sect:stability} and \ref{sect:II}) separately in Sections \ref{sect:stability} and \ref{sect:II} for some constant $\rblue\in (r_-, r_+)$,  both of which are the  so-called {\it red-shift } and {\it blue-shift} subregions of $\MM$. Hence, in the context, whenever we write partial derivatives $\pu$ and $\pv$, it is clear of what they mean depending on which region we are considering. 

We relate the coordinates $(u,\ub,\th,\phiin)$ with the ingoing E--F coordinates $(\ubin, \rin, \thin,\phiin)$  by
\begin{align}
u=2r^*(\rin) - \ubin, \quad \ub=\ubin, \quad \th=\thin.
\end{align}
Hence, the right event horizon $\Horizon\backslash \Sphere_{\HH}$, on which $\rin=r_+$ and $r^*(\rin)=-\infty$, can be identified as  the limiting hypersurface $\{(u, \ub, \th, \phiin)\vert u=-\infty\}$. In the same manner, the left event horizon $\Horizon'\backslash \Sphere_{\HH}$, the right Cauchy horizon $\CHorizon\backslash \Sphere_{\CHH}$ and the left Cauchy horizon $\CHorizon'\backslash \Sphere_{\CHH}$  are identified as the limiting hypersurface $\{(u, \ub, \th, \phiout)\vert \ub=-\infty\}$,   the limiting hypersurface $\{(u, \ub, \th, \phiout)\vert \ub=+\infty\}$ and the limiting hypersurface $\{(u, \ub, \th, \phiin)\vert u=+\infty\}$, respectively.

\subsection{Scalar wave equation in Kerr}
\label{subsect:scalar:Kerr:form}

The scalar wave equation \eqref{eq:wave}  in the Boyer--Lindquist coordinates takes the following form:
\begin{align}
&  -\frac{(r^2+a^2)^2-a^2 \sin^2{\theta} \Delta}{\Delta}
\partial_{tt}^2\psi
+\partial_r( \Delta\partial_r\psi)
- \frac{4aMr}{\Delta}\partial_{t\phi}^2 \psi-\frac{a^2}{\Delta}\partial_{\phi\phi}^2\psi+\Delta_{\Sphere}\psi=0,
\end{align}
where $\Delta_{\Sphere}$ is the covariant spherical Laplacian on unit $2$-sphere and $\Delta_{\Sphere}\psi= \frac{1}{\sin{\theta}} \partial_{\theta}( \sin{\theta} \partial_{\theta}\psi)
+\frac{1}{\sin^2{\theta}}\partial_{\phi\phi}^2\psi$ in the Boyer-Lindquist coordinates.

In view of \eqref{def:e3e4:1}, the scalar wave equation can be rewritten as
\begin{align}
\label{eq:scalar:form:e3e4}
  -4(r^2+a^2)e_3e_4\psi+\Delta_{\Sphere} \psi+a^2 \sin^2 \theta\Lxi^2\psi
  +2a\Lxi\Leta\psi
  +2r(e_4-\mu e_3)\psi
  +\frac{2ar}{r^2+a^2}\Leta\psi=0.
\end{align}

\subsection{Mode decomposition}
\label{subsect:modeproj}

Let $\dSp$ be the volume element of the standard unit $2$-sphere $\Sphere$. At any point $(u,\ub,\omega)$ in $\MM$, for a real-valued function $\varphi$ on $\MM$, we define
\begin{align}
\varphi_{\ell=0}=\varphi_{\ell=0}(u,\ub)\doteq \frac{1}{4\pi}\int_{\Sphere(u,\ub)} \varphi\ \dSp, \qquad \varphi_{\ell\geq 1}=\varphi_{\ell\geq 1}(u,\ub,w)\doteq \varphi -\varphi_{\ell=0}.
\end{align}
These are the $\ell=0$ mode (i.e., the spherically symmetric part) and the $\ell\geq 1$ modes of the scalar function $\varphi$, respectively. Further, we define two mode projection operators by
\begin{align}
\PJ_{\ell=0}(\varphi)\doteq \varphi_{\ell=0},\qquad \PJ_{\ell\geq 1}(\varphi)\doteq \varphi_{\ell\geq 1}.
\end{align}

\section{Precise late-time asymptotics in the region $\DD\cap\{\rblue\leq r\leq r_+\}$}
\label{sect:stability}

 In this section, we derive the energy decay estimates and prove the precise late-time asymptotic estimates for the scalar field in the region $\I=\DD\cap\{\rblue\leq r\leq r_+\}\cap\{\ub\geq 1\}$ close to the right event horizon $\HH_{+}$, as well as in the region $\I'=\DD\cap\{\rblue\leq r\leq r_+\}\cap\{u\geq 1\}$ close to the left event horizon $\HH_+'$, under the assumed estimates \eqref{ass:thm:horizon:psi} on the event horizon. Note that $\rblue$ is a constant in $(r_-, r_+)$ and will be fixed in Section \ref{sect:II}. 
 
 Our main statement on the precise asymptotics in this region is included in the following theorem.
 
\begin{thm}[Precise late-time asymptotics in the region $\DD\cap\{\rblue\leq r\leq r_+\}$]
\label{thm:regionI}
Under the assumptions in Theorem \ref{thm:main}, we have
\begin{itemize}
\item in region $\I=\DD\cap\{\rblue\leq r\leq r_+\}\cap\{\ub\geq 1\}$,
 \begin{subequations}
 \label{eq:main:ptw:I}
  \begin{align}
  \label{eq:main:ptw:I:psi}
    |\psi-c_0\ub^{-3}| &\lesssim \ub^{-3-\epsilon},\\
     \label{eq:main:ptw:I:pubpsi}
    \bigg|e_4\psi+\frac{3}{2}c_0\bigg(1+\frac{r_+^2+a^2}{\R}\bigg)\ub^{-4}\bigg| &\lesssim \ub^{-4-\epsilon},\\
     \label{eq:main:ptw:I:pupsi}
    \bigg|e_3\psi+\frac{3}{2}c_0\frac{r+r_+}{r-r_-}\ub^{-4}\bigg| &\lesssim \ub^{-4-\epsilon};
  \end{align}
  \end{subequations}
\item and in region $\I'=\DD\cap\{\rblue\leq r\leq r_+\}\cap\{u\geq 1\}$,
  \begin{subequations}
 \label{eq:main:ptw:I'}
  \begin{align}
  \label{eq:main:ptw:I':psi}
    |\psi-c_0' u^{-3}| &\lesssim u^{-3-\epsilon},\\
     \label{eq:main:ptw:I':pupsi}
    \bigg|e_3'\psi+\frac{3}{2}c_0'\bigg(1+\frac{r_+^2+a^2}{\R}\bigg)u^{-4}\bigg| &\lesssim u^{-4-\epsilon},\\
     \label{eq:main:ptw:I':pvpsi}
    \bigg|e_4'\psi+\frac{3}{2}c_0'\frac{r+r_+}{r-r_-}u^{-4}\bigg| &\lesssim u^{-4-\epsilon}.
  \end{align}
  \end{subequations}
  \end{itemize}
\end{thm}

We shall from now on consider only the region $\I$, since the proof of the other region $\I'$ follows in an exactly same manner by making the replacements $u\rightarrow \ub$, $\ub\rightarrow u$, $c_0\rightarrow c_0'$, $e_3\rightarrow e_4'$, and $e_4\rightarrow e_3'$. 

\textbf{Recall from Section \ref{subsect:almostnull} that in this region $\I$, we use the coordinates $(u,\ub, \th,\phiin)$, and
the derivatives in this coordinate system are expressed in \eqref{eq:deris:uub:nullframe:I}.}  In particular, $e_3=-\mu^{-1}\pu$ and $e_4=\pv+\frac{a}{\R}\Leta$.

 \subsection{Energy decay for scalar field in region $\I$}
 \label{sect:ED:scalar:I}

Multiplying the scalar wave equation \eqref{eq:scalar:form:e3e4} by $\mu$ on both sides, it becomes
\begin{align}\label{eq-wave-doublenull}
  4(r^2+a^2)\pu e_4(\psi)+\mu\Delta_{\Sphere}\psi
  +2\mu a\Lxi\Leta\psi+\mu a^2 \sin^2 \theta\Lxi^2\psi
  +2\mu r(\pu\psi+ e_4(\psi))
  +\frac{2\mu ar}{r^2+a^2}\Leta\psi=0.
\end{align}

We then multiply on both sides of \eqref{eq-wave-doublenull} by $X(\psi)\doteq f(r)\pu\psi+g(r)e_4(\psi)$ and integrate over $\Sp$, and we get the following energy identity:
 \begin{align}\label{eq:div:fullfield}
   \pv\biggl(\int_{\Sp}\fv[\psi]\,\dSp\biggr)+\pu\biggl(\int_{\Sp}\fu[\psi]\,\dSp\biggr) +\int_{\Sp}\mathcal{B}[\psi]\,\dSp =0,
 \end{align}
 where
 \begin{subequations}
 \begin{align}
    \label{exp:fv}
   \fv[\psi]&=2(r^2+a^2)f(r)|\pu\psi|^2-\half \mu g(r) \big((e_1(\psi))^2 + (e_2(\psi))^2\big)
   +a\mu\sin\theta e_2(\psi) X(\psi),\\
   \label{exp:fu}
   \fu[\psi]&=2(r^2+a^2)g(r)|e_4(\psi)|^2-\half \mu f(r) \big((e_1(\psi))^2 + (e_2(\psi))^2\big)
-a\mu\sin\theta e_2(\psi) X(\psi),
 \end{align}
 and
 \begin{align}\label{energy:spacetime-term}
   \mathcal{B}[\psi] ={}&2\big(r\mu g(r)-\pu((r^2+a^2)g(r))\big)|e_4(\psi)|^2+2\big(r\mu f(r)-e_4((r^2+a^2)f(r))\big)|\pu\psi|^2\notag\\
   &+\half[\pu(\mu f(r))+e_4(\mu g(r))]\big((e_1(\psi))^2 + (e_2(\psi))^2\big)\notag\\
   &+2\mu r(g(r)+f(r))e_4(\psi) \pu\psi
 +\frac{2\mu ar}{r^2+a^2} 
  \Leta\psi (-f(r)\pu\psi+g(r)e_4(\psi)).
 \end{align}
 \end{subequations}
 These computations are most convenient to perform by multiplying $f(r)\pu\psi$ and $g(r)e_4(\psi)$ on both sides of
 \begin{subequations}
 \label{eq:scalar:equivalent:2}
 \begin{align}
  &4(r^2+a^2)\pu e_4(\psi)+\mu
  \Big(\frac{e_1 (\sin\th e_1 (\psi))}{\sin\th}
  + e_2e_2(\psi)\Big)
  +2\mu r(\pu\psi+ e_4(\psi))
  +\frac{2\mu ar}{r^2+a^2}\Leta\psi=0\\
  &4(r^2+a^2)e_4\pu (\psi)+\mu
  \Big(\frac{e_1 (\sin\th e_1 (\psi))}{\sin\th}
  + e_2e_2(\psi)\Big)
  +2\mu r(\pu\psi+ e_4(\psi))
  -\frac{2\mu ar}{r^2+a^2}\Leta\psi=0
\end{align}
\end{subequations}
respectively, both of the above subequations in \eqref{eq:scalar:equivalent:2} being alternative forms of equation \eqref{eq-wave-doublenull}.

\begin{prop}\label{lemma:decay:estimate:full:red}
Under the assumed estimates \eqref{ass:thm:horizon:psi} in Theorem \ref{thm:main}, we have for any  $1\leq \tub_1<\tub_2$, 
  \begin{align}\label{energy:decay:redshift-region:full}
    \iint\limits_{\{\tub=\tub_1\}\cap\I}\Enn[\Lxi^j\psi](-\mu)\dSp\di u\lesssim \tub_1^{-8-2j},
  \end{align}
  and
  \begin{align}\label{energy:decay:red:spacetime-esti}
    \iint\limits_{\{r=\rblue\}\cap\{\tub_1\leq \tub\leq \tub_2\}}\Enn[\Lxi^j\psi]\dSp\di \ub+
    \iint\limits_{\{\tub_1\leq \tub\leq \tub_2\}\cap\I} \Enn[\Lxi^j\psi](-\mu)\dSp\di u\di \ub
    \lesssim
   \tub_1^{-8-2j}+\int_{\tub_1}^{\tub_2}\tub^{-8-2j}\di \tub,
  \end{align}
  where 
  \begin{align}
    \Enn[\psi]&=\sum_{j=1}^4 \abs{e_j(\psi)}^2.
  \end{align}
\end{prop}

\begin{proof} 
It suffices to prove $j=0$ case since $\Lxi$ commutes with the scalar wave equation. 

By integrating \eqref{eq:div:fullfield} in the region $\{(u,\ub)| \tub_1\leq \tub(u,\ub)\leq \tub_2 \}\cap \I$ with respect to $\di u \di\ub$,  we have
  \begin{align}
    &\iint\limits_{\{\tub=\tub_2\}\cap\I}\mathbb{T}_{\tub}[\psi]\dSp\di u
   + \iint\limits_{\{r=\rblue\}\cap\{\tub_1\leq \tub\leq \tub_2\}}\mathbb{T}_{r}[\psi]\dSp\di\ub
    +\iint\limits_{\{\tub_1\leq \tub\leq \tub_2\}\cap\I}\mathcal{B}[\psi]\dSp\di u\di\ub \notag\\
    =&\iint\limits_{\{\tub=\tub_1\}\cap\I}\mathbb{T}_{\tub}[\psi]\dSp\di u
    +\iint\limits_{\Horizon\cap \{\tub_1\leq \tub\leq \tub_2\}}\mathbb{T}_{u}[\psi]\dSp\di \ub,
  \end{align}
  where $\Entub=(1-\half\mu)\fv[\psi]-\half\mu\fu[\psi]$, $\Enr=-\half\mu (\fv[\psi]+\fu[\psi])$ and $\mathbb{T}_{u}[\psi]=\fu[\psi]$. 
We choose
 \begin{align}\label{test-functions}
   f(r)=(r^2+a^2)^p(-\mu)^{-1},\ \ g(r)=(r^2+a^2)^p,
 \end{align} 
with $p$ being a positive number to be fixed.

 Given this choice of functions $f$ and $g$, we calculate the coefficients in $\mathcal{B}[\psi]$ as follows:
 \begin{subequations}
 \begin{align}
    r\mu g(r)-\pu((r^2+a^2)g(r))&=-\mu pr(r^2+a^2)^p,\\
    r\mu f(r)-e_4((r^2+a^2)f(r))&=(-\mu)^{-1}
    [(r-M)
    -\mu(p+1)r](r^2+a^2)^p,\\
    \pu(\mu f(r))+e_4(\mu g(r))&=-\mu(r^2+a^2)^{p-1}\bigg(\frac{2Mpr^2}{r^2+a^2}+r\mu-(r-M)\bigg).
\end{align}
\end{subequations}
Since $(r-M)
-\mu(p+1)r\geq (r-M)+\frac{(p+1)r_-(\rblue-r_-)}{r_+^2+a^2}(r_+-r)$ {in} $\I$,  we have $(r-M)
+(p+1)r(-\mu)\gtrsim - \mu p +1$ in $\I$. By taking $p$ large, this implies that the first two lines of \eqref{energy:spacetime-term} are bounded below by
\begin{align}
 c(\R)^p \big( p \big[(-\mu)  \big(|e_4(\psi)|^2+(e_1(\psi))^2 + (e_2(\psi))^2\big)\big]+(p+(-\mu)^{-1})|\pu\psi|^2\big),
\end{align}  
for a constant $c>0$. 
We then use the following formula 
\begin{align}
 \frac{q^2}{\R} \Leta=\sin\th e_2-a\sin^2\th \Big(\Lxi+\frac{a}{\R}\Leta\Big)
  =\sin\th e_2-a\sin^2\th e_4 +a\sin^2\th \pu
\end{align}
to expand the last line of \eqref{energy:spacetime-term} and apply 
 Cauchy--Schwarz inequality to the last line, and one finds the absolute value of the last line of \eqref{energy:spacetime-term} is bounded by
 \begin{align}
C(\R)^p \big[\veps_0((-\mu)^{-1} +1)\abs{\pu \psi}^2 - \veps_0^{-1} \mu (\abs{e_4(\psi)}^2 + \abs{e_2 (\psi)}^2)
\big].
 \end{align}
 By taking $\veps_0$ small and then $p$ sufficiently large (compared to $\veps_0^{-1}$), the last line is absorbed by the first two lines of \eqref{energy:spacetime-term}. Consequently, we achieve
 \begin{align}
  \mathcal{B} [\psi]\gtrsim {}&  (-\mu)  \Enn[\psi], \quad \text{in}\ \I.
\end{align}

Now fix the value of $p$, and consider the boundary terms.  Observe that
 \begin{align}
   \Entub ={}&(1-\half\mu)\fv[\psi]-\half\mu\fu[\psi]\notag\\
   ={}&2(r^2+a^2)f(r)|\pu\psi|^2-\half \mu g(r) \big((e_1(\psi))^2 + (e_2(\psi))^2\big)
   +a\mu\sin\theta e_2(\psi) X(\psi)
   \notag\\
   &+(-\mu)\Big[(r^2+a^2)\big(f(r)|\pu\psi|^2+g(r)|e_4(\psi)|^2\big)+ \frac{1}{4} (-\mu) (f(r)+g(r) )\big((e_1(\psi))^2 + (e_2(\psi))^2\big)\Big]\notag\\
   \sim {}&
   (-\mu)\big(|e_4(\psi)|^2+\abs{\mu^{-1}\pu\psi}^2+(e_1(\psi))^2 + (e_2(\psi))^2\big)\notag\\
   \sim {}& (-\mu)   \Enn[\psi], \qquad \text{in} \,\,\, \I,\\
   \Enr={}&-\half\mu (\fv[\psi]+\fu[\psi])\notag\\
   ={}&(-\mu)\Big[(r^2+a^2)f(r)|\pu\psi|^2+(r^2+a^2)g(r)|e_4(\psi)|^2- \frac{1}{4} \mu (f(r)+g(r) )\big((e_1(\psi))^2 + (e_2(\psi))^2\big)\Big]\notag\\
   \sim {} & \Enn[\psi],\qquad \text{on}\ \{r=\rblue\}.
 \end{align}
 For the boundary term on $\Horizon$, since $\mathbb{T}_{u}[\psi]=\fu[\psi]$, one finds from \eqref{exp:fu} that 
 \begin{align}
\bigg| \iint\limits_{\Horizon\cap \{\tub_1\leq \tub\leq \tub_2\}}\mathbb{T}_{u}[\psi]\dSp\di \ub\bigg|
\lesssim{}&\bigg| \iint\limits_{\Horizon\cap \{\tub_1\leq \tub\leq \tub_2\}}\big(\abs{e_4(\psi)}^2 +(e_1(\psi))^2 + (e_2(\psi))^2\big) \dSp\di \ub\bigg|\notag\\
\lesssim{}&\int_{\tub_1}^{\tub_2}\tub^{-8}\di \tub.
 \end{align}
 Here, we have used the fact that when restricted on $\Horizon$,
 \begin{align}
 \hspace{4ex}&\hspace{-4ex}\abs{e_4(\psi)}^2 +(e_1(\psi))^2 + (e_2(\psi))^2\notag\\
 \lesssim{}&\big(\abs{e_4(\psi_{\ell=0})}^2 +(e_1(\psi_{\ell=0}))^2 + (e_2(\psi_{\ell=0}))^2\big) 
 +\big(\abs{e_4(\psi_{\ell\geq 1})}^2 +(e_1(\psi_{\ell\geq 1}))^2 + (e_2(\psi_{\ell\geq 1}))^2\big) 
 \notag\\
 \lesssim {}&\abs{\Lxi (\psi_{\ell=0})}^2 +\big(\abs{e_4(\psi_{\ell\geq 1})}^2 +(e_1(\psi_{\ell\geq 1}))^2 + (e_2(\psi_{\ell\geq 1}))^2\big) \notag\\
 \lesssim {}& \ub^{-8}\lesssim \tub^{-8},
 \end{align} 
 by the assumption \eqref{ass:thm:horizon:psi} in Theorem \ref{thm:main} and by the formula
$\ub =\tub$ on $\Horizon$.

The above discussions together thus yield that
 \begin{align}
    &\iint\limits_{\{\tub=\tub_2\}\cap\I}\Enn[\psi](-\mu)\dSp\di u
   + \iint\limits_{\{r=\rblue\}\cap\{\tub_1\leq \tub\leq \tub_2\}}\Enn[\psi]\dSp\di\ub
    +\iint\limits_{\{\tub_1\leq \tub\leq \tub_2\}\cap\I}\Enn[\psi](-\mu)\dSp\di u\di\ub \notag\\
    \lesssim{}&\iint\limits_{\{\tub=\tub_1\}\cap\I}\Enn[\psi](-\mu)\dSp\di u
    +\int_{\tub_1}^{\tub_2}\tub^{-8}\di \tub.
  \end{align}
In the end, the energy decay estimate \eqref{energy:decay:redshift-region:full} follows by applying the following Lemma \ref{lemma:mean:decay} with $\alpha=0$ to the above inequality, and  \eqref{energy:decay:red:spacetime-esti} follows from the above inequality and the estimate \eqref{energy:decay:redshift-region:full}.
\end{proof}
\begin{remark}\label{remark:boundedness:rblue}
	As a consequence of the proof, we can also obtain the boundedness of the energy on the spacelike hypersurface ${r=\rblue}$, given by
	\begin{align}
		\iint_{r=\rblue}\Enn[\psi]\lesssim 1.
	\end{align}
 In fact, this statement follows from  combining the estimates in $\I$ and $\I'$, along with a standard estimate in the remaining region $\DD\cap\{\ub\leq 1\}\cap\{u\leq 1\}$. 
\end{remark}

\begin{lemma}\label{lemma:mean:decay}
 Let $p>1$ and $0\leq \alpha<1$, and let  $f(x): [1,+\infty)\to [0,+\infty)$ be a continuous function. Assume that there are  constants $C_0>0$, $C_1>0$, $C_2\geq 0$ and $C_3\geq 0$ such that  for any $1\leq  x_1<x_2$, 
    \begin{align}\label{function:property:3}
      f(x_2)+C_1\int_{x_1}^{x_2}x^{-\alpha}f(x)\di x\leq C_0  f(x_1)+C_2\int_{x_1}^{x_2}x^{-p}\di x +C_3 x_1^{-p}.
    \end{align}
    Then, we have, for any $x_1\geq 1$,
    \begin{align}
    \label{function:decay:1} 
      f(x_1)  \leq {}&  C x_1^{-p+\alpha},
    \end{align}
    where $C$ is a constant depending only on $f(1)$, $C_0, C_1,C_2, C_3$, $p$ and $\alpha$.
 \end{lemma}

 \begin{proof}
 Using \eqref{function:property:3}, we get, for $x\geq 1$, 
   \begin{align}
   \label{eq:meanvaluelemma:step1:943}
    f(x)\leq C_0 f(1) + C_2 \int_{1}^{x} (x')^{-p} \di x' + C_3\leq C_0 f(1)+\frac{C_2}{p-1} +C_3.
   \end{align}  
   
   To prove \eqref{function:decay:1}, we define
   \begin{align}
   \label{def:g(x)}
     g(x):=\max\{f(x)-\frac{C_2}{C_1}x^{-p+\alpha},0\}.
   \end{align}
   Take an arbitrary $x\geq 1$.  Let $\palpha$ be the minimal integer greater than or equal to $\frac{p}{1-\alpha}$. If $g(2^{\palpha}x)=0$, then we have $f(2^{\palpha}x)\leq \frac{C_2}{C_1}(2^{\palpha}x)^{-p+\alpha}$. 
Otherwise, it suffices to consider $g(2^{\palpha}x)>0$ case, and we divide it into the following two subcases. 

The first subcase is $g(y)>0$ for all $y\in[x,2^{\palpha}x]$. The definition of $g$ and inequality \eqref{function:property:3} yield
   \begin{align}\label{lemma:function:property:3}
    g(x_2)+C_1\int_{x_1}^{x_2}x^{-\alpha}g(x)\di x\leq C_0 g(x_1)+\tilde{C}_3 x_1^{-p+\alpha}
   \end{align}
   for any $x\leq  x_1<x_2\leq 2^{\palpha}x$, with $\tilde{C}_3=\frac{C_0C_2}{C_1} +C_3$.
   By the mean-value principle,  there exists an $\tilde{x}\in [x,2x]$, such that $x\tilde{x}^{-\alpha} g(\tilde{x})=\int_{x}^{2x}x^{-\alpha}g(x')\di x'$. Hence, by \eqref{lemma:function:property:3}, we have $2^{-\alpha}C_1 g(\tilde{x})\leq C_0x^{-1+\alpha}g(x)+\tilde{C}_3 x^{-p-1+2\alpha}$ and $g(2x)\leq C_0g(\tilde{x})+\tilde{C}_3 \tilde{x}^{-p+\alpha}$, which means
   \begin{align}
     g(2x)\leq \frac{2^\alpha C_0^2}{C_1}x^{-1+\alpha}g(x)+\bigg(\frac{2^\alpha C_0}{C_1}+1\bigg)\tilde{C}_3 x^{-p+\alpha},
   \end{align}
   where we have used $-1+2\alpha<\alpha$ for $0\leq \alpha<1$. 
   Iterating the above inequality $\palpha$ times, we obtain
   \begin{align}
     g(2^{\palpha}x)\leq C\big(x^{-{\palpha} (1-\alpha)}g(x)+ x^{-p+\alpha}\big)\leq C(x^{-p} g(x) + x^{-p+\alpha}),
   \end{align}
   with $C$ being a constant depending only on $\{C_j\}_{j=0,1,2,3}$, $p$ and $\alpha$.
   Together with the definition \eqref{def:g(x)} of $g(x)$ and the estimate \eqref{eq:meanvaluelemma:step1:943} of $f(x)$, we derive $g(2^{\palpha}x)\leq C (2^{\palpha}x)^{-p+\alpha}$,  and hence $f(2^{\palpha} x)\leq C (2^{\palpha} x)^{-p+\alpha}$, with $C$ being a constant depending only on $\{C_j\}_{j=0,1,2,3}$,  $p$ and $\alpha$. 
   
The remaining subcase is that there exists a $\bar{x}\in [x,2^{\palpha}x)$ such that $g(\bar{x})=0$ and $g(y)>0$ for $y\in(\bar{x},2^{\palpha}x]$. This yields $f(\bar{x})=\frac{C_2}{C_1}\bar{x}^{-p+\alpha}$ and $f(y)>\frac{C_2}{C_1}y^{-p+\alpha}$ for $y\in(\bar{x},2^{\palpha}x]$. Then, because of this fact, we use \eqref{function:property:3} with $x_2=2^{\palpha} x$ and $x_1=\bar{x}$ and deduce
   \begin{align}
    f(2^{\palpha}x)\leq C_0f(\bar{x})+ C_3\bar{x}^{-p}\leq \bigg(\frac{C_0C_2}{C_1}+C_3\bigg)x^{-p+\alpha}.
   \end{align}
   
   Since $x\geq 1$ is arbitrary, we have proved from the above that for any $x\in[2^{\palpha},+\infty)$, the estimate \eqref{function:decay:1} holds true. The estimate \eqref{function:decay:1} for $x\in [1, 2^{\palpha}]$ follows easily from \eqref{eq:meanvaluelemma:step1:943}. 
   \end{proof}
 
\subsection{Energy decay for $\psi_{\ell=0}$ and $\psi_{\ell\geqslant 1}$ in $\I$}

By projecting the wave equation \eqref{eq-wave-doublenull} onto $\ell\geq 1$ modes, we achieve
\begin{align}\label{eq-wave-doublenull-mode}
 \hspace{4ex}&\hspace{-4ex} 
 4(r^2+a^2)\pu e_4(\psilarge)
  +2a\mu\Lxi\Leta(\psilarge)
+\mu\Delta_{\Sphere}(\psilarge)
+2\mu r(\pu(\psilarge)+ e_4(\psilarge))
+\frac{2ar\mu}{r^2+a^2}\Leta(\psilarge)\notag\\
  ={}&-\mu a^2\PJ_{\ell\geq 1} (\sin^2\theta  \Lxi^2\psi)\notag\\
  ={}&-\mu a^2 \sin^2\theta  \Lxi^2(\psilarge)
  -\mu a^2[\PJ_{\ell\geq 1}, \sin^2\theta]  (\Lxi^2\psi).
\end{align}

Instead, if we project equation \eqref{eq-wave-doublenull} onto the $\ell=0$ mode, we obtain the same equation as \eqref{eq-wave-doublenull-mode} but replacing $\ell\geq 1$ by $\ell=0$ everywhere. Since $\Leta\psiz=\Delta_{\Sphere}(\psiz)=0$, one can further simply its equation to
\begin{align}\label{eq-wave-doublenull-mode:=0}
 \hspace{4ex}&\hspace{-4ex} 
 4(r^2+a^2)\pu \pv(\psiz)
+2\mu r(\pu(\psiz)+ \pv(\psiz))+\mu a^2 \sin^2\theta  \Lxi^2(\psiz)
\notag\\
  ={}&
  -\mu a^2[\PJ_{\ell=0}, \sin^2\theta]  (\Lxi^2\psi).
\end{align}

\begin{lemma}\label{Lemma:mode:red}
  Under the assumed estimates \eqref{ass:thm:horizon:psi} in Theorem \ref{thm:main}, we have for any $\ub_1\geq 1$ that
  \begin{align}\label{energy:decay:redshift-region:zeromode}
    \int_{\{\ub=\ub_1\}\cap\{\rblue\leq r<r_+\}}(-\mu)^{-1}|\pu\Lxi^j\psiz|^2\di u\lesssim \ub_1^{-8-2j},
  \end{align}
  and, for any $\tub_1\geq 1$ and  $k_1\geq 0$, $k_2\geq0$, 
  \begin{align}\label{energy:decay:redshift-region:highermode}
    \iint\limits_{\{\tub=\tub_1\}\cap\I}\Enn[\Lxi^j(\Carter^{k_2}\Leta^{k_1}\psi)_{\ell\geq 1}](-\mu)\dSp\di u\lesssim \tub_1^{-8-2j-2\delta},
  \end{align}
  where $\delta$ is given in \eqref{assump:eq:HH:highmodes}.
\end{lemma}

\begin{proof}
The proof is the same as the one of Lemma \ref{lemma:decay:estimate:full:red} by taking the multiplier $X(\psilarge):=f(r)\pu(\psilarge)+g(r)e_4(\psilarge)$ with the same choice of functions $f(r)$ and $g(r)$ and integrating over $\{(u,\ub)\times \Sphere(u,\ub) |  \tub_1\leq \tub\leq \tub_2\}$. For the same reason, we only consider $j=0$ case. Moreover, since both the Carter tensor $\Carter$ and the vector field $\Leta$ commute with the wave operator, it suffices to consider $k_1=k_2=0$. To complete the proof, it suffices to control the error term arising from the last term of equation \eqref{eq-wave-doublenull-mode} by
\begin{align}
&
\bigg| \iint\limits_{\{\tub_1\leq \tub\leq \tub_2\}\cap\I}\mu a^2[\PJ_{\ell\geq 1}, \sin^2\theta]  (\Lxi^2\psi) \times X(\psilarge)\dSp\di u\di \ub\bigg|\notag\\
&\lesssim \iint\limits_{\{\tub_1\leq \tub\leq \tub_2\}\cap\I}\Big(\veps \Enn[\psilarge] +\veps^{-1}\big|a^2[\PJ_{\ell\geq 1}, \sin^2\theta]  (\Lxi^2\psi) \big|^2\Big)
(-\mu)\dSp\di u\di \ub\notag\\
&\lesssim \veps\iint\limits_{\{\tub_1\leq \tub\leq \tub_2\}\cap\I} \Enn[\psilarge] (-\mu)\dSp\di u\di \ub
+\veps^{-1}\iint\limits_{\{\tub_1\leq \tub\leq \tub_2\}\cap\I} \Enn[\Lxi\psi](-\mu)\dSp\di u\di \ub\notag\\
&\lesssim \veps\iint\limits_{\{\tub_1\leq \tub\leq \tub_2\}\cap\I} \Enn[\psilarge] (-\mu)\dSp\di u\di \ub
+\veps^{-1}\bigg(\tub_1^{-10}+\int_{\tub_1}^{\tub_2}\tub^{-10}\di \tub\bigg),
\end{align} 
where we have used the estimate \eqref{energy:decay:red:spacetime-esti} in the last step. By taking $\veps$ small enough such that the first term in the last line is absorbed, the desired estimate \eqref{energy:decay:redshift-region:highermode} is proved.

For the mode $\psiz$, the proof is the same but integrating in the region $\{(u,\ub)\times \Sphere(u,\ub) | \ub_1\leq \ub\leq \ub_2\}$, and we omit it.
\end{proof}

\subsection{Precise late-time asymptotics in region $\I$}

We now prove the estimates \eqref{eq:main:ptw:I} of Theorem \ref{thm:main} under the assumed estimates \eqref{ass:thm:horizon:psi}.

Prior to the proof, a lemma on Sobolev-imbedding estimates is necessary. 

\begin{lemma}
  We have\footnote{For a functional $G: X\to Y$, where $X, Y$ are function spaces, we use the notation $G(P^{\leq k}u) $  to denote the sum $\sum_{j=0}^k G(P^ju)$, where $u\in X$ and $P$ is an operator from $X\to X$.}
  \begin{subequations}\label{SI:estimate:sph}
  \begin{align}
    |\psih|_{L^\infty(\Sphere_{u,\ub})}^2 \lesssim &\int_{\Sphere_{u,\ub}}(\Endeg[(\Carter^{\leq 1}\psi)_{\ell\geq 1}]+\Endeg[T\psi])\dSp, \label{SI:estimate:sph:1}\\
    |\mu e_3\psih|_{L^\infty(\Sphere_{u,\ub})}^2 +|e_4\psih|_{L^\infty(\Sphere_{u,\ub})}^2 \lesssim & \int_{\Sphere_{u,\ub}}(\Endeg[(\Carter^{\leq 1}\psi)_{\ell\geq 1}]+\Endeg[T^2\psi])\dSp,\label{SI:estimate:sph:2}\\
    |e_1\psih|_{L^\infty(\Sphere_{u,\ub})}^2 +|e_2\psih|_{L^\infty(\Sphere_{u,\ub})}^2 \lesssim & \int_{\Sphere_{u,\ub}}(\Endeg[\Lxi^{\leq 1}(\Carter^{\leq 1}\psi)_{\ell\geq 1}]+\Endeg[T^{\leq 1}\Lxi\psi])\dSp,\label{SI:estimate:sph:3}
  \end{align}
\end{subequations}
  where $\Endeg[\psi]=|e_1(\psi)|^2+|e_2(\psi)|^2+|\mu e_3\psi|^2+|e_4\psi|^2$.
  \end{lemma}
  
  \begin{proof}
  Recall the definition $\Carter=\Delta_{\Sphere}+a^2\sin^2\th \Lxi^2$. Hence we have
  \begin{align}
    (\Carter\psi)_{\ell\geq 1}=(\Delta_{\Sphere}\psi+a^2\sin^2\th \Lxi^2\psi)_{\ell\geq 1}=\Delta_{\Sphere}(\psi_{\ell\geq 1})+(a^2\sin^2\th \Lxi^2\psi)_{\ell\geq 1},
  \end{align}
  where we have used $(\Delta_{\Sphere}\psi)_{\ell\geq 1}=\Delta_{\Sphere}\psi=\Delta_{\Sphere}(\psiz+\psih)=\Delta_{\Sphere}(\psih)$. 
  Then, by a standard Sobolev-imbedding estimate on sphere, we deduce
  \begin{align}
    |\psih|_{L^\infty(S_{u,\ub}^2)}^2 &\lesssim 
    \int_{\Sphere_{u,\ub}}(|\Delta_{\Sphere}\psih|^2+|\psih|^2)\dSp\notag\\
  &\lesssim \int_{\Sphere_{u,\ub}}(|(\Carter\psi)_{\ell\geq 1}|^2+|(a^2\sin^2\th \Lxi^2\psi)_{\ell\geq 1}|^2+|\psih|^2)\dSp\notag\\
  &\lesssim \int_{\Sphere_{u,\ub}}(|(\Carter\psi)_{\ell\geq 1}|^2+|a^2\sin^2\th \Lxi^2\psi|^2+|\psih|^2)\dSp\notag\\
  &\lesssim \int_{\Sphere_{u,\ub}}(|\slashed{\nabla}(\Carter\psi)_{\ell\geq 1}|^2+|a|^2|a\sin\th \Lxi^2\psi|^2+|\slashed{\nabla}\psih|^2)\dSp\notag\\
  &\lesssim \int_{\Sphere_{u,\ub}}\bigg(|\slashed{\nabla}(\Carter^{\leq 1}\psi)_{\ell\geq 1}|^2+|e_2 \Lxi\psi|^2+\bigg|\frac{1}{\sin\th}\Leta \Lxi\psi\bigg|^2\bigg)\dSp,
  \end{align}
where we have used the Poincar\'e inequality in the second line line for $(\Carter\psi)_{\ell\geq 1}$ and $\psih$ which are  not supported on $\ell=0$ mode. 
 The estimate \eqref{SI:estimate:sph:1} follows from the above inequality by substituting  $|\slashed{\nabla}\psi|^2=|\partial_\th\psi|^2+\frac{1}{\sin^2\theta}|\Leta\psi|^2$ and $\frac{q^2}{\R} \Leta\psi
  =\sin\th e_2(\psi)-a\sin^2\th (e_4\psi +\mu e_3\psi)$ into the last line.

  By $[\mu e_3, \Carter]=[e_4,\Carter]=0$, one can easily derive 
  \begin{align}\label{proof:sb:im:234}
    |\mu e_3\psih|_{L^\infty(\Sphere_{u,\ub})}^2+|e_4\psih|_{L^\infty(\Sphere_{u,\ub})}^2
    \lesssim &\int_{\Sphere_{u,\ub}}(|\mu e_3(\Carter^{\leq 1}\psi)_{\ell\geq 1}|^2+|\mu e_3\Lxi^2\psi|^2)\dSp\notag\\
    &+\int_{\Sphere_{u,\ub}}(| e_4(\Carter^{\leq 1}\psi)_{\ell\geq 1}|^2+| e_4\Lxi^2\psi|^2)\dSp,
  \end{align}
and this implies \eqref{SI:estimate:sph:2}.

  In the end, note that 
  \begin{align}
    |\slashed{\nabla}\psi|_{L^\infty(\Sphere_{u,\ub})}^2 &\lesssim 
    \int_{\Sphere_{u,\ub}}(|\slashed{\nabla}(\Delta_{\Sphere}\psi)|^2+|\slashed{\nabla}\psi|^2)\dSp\notag\\
    &\lesssim 
    \int_{\Sphere_{u,\ub}}(|\slashed{\nabla}(\Carter\psi)|^2+|\slashed{\nabla}(\sin^2\th \Lxi^2\psi)|^2+|\slashed{\nabla}\psi|^2)\dSp\notag\\
    &\lesssim 
    \int_{\Sphere_{u,\ub}}(|\slashed{\nabla}(\Carter^{\leq 1}\psi)_{\ell\geq 1}|^2+|\slashed{\nabla}^{\leq1}\Lxi^2\psi)|^2)\dSp.
  \end{align}
  This together with \eqref{SI:estimate:sph:1} and $|e_1\psih|^2+|e_2\psih|^2\lesssim |\slashed{\nabla}\psi|^2+|\Lxi\psih|^2$ yields \eqref{SI:estimate:sph:3}.
  \end{proof}

\begin{proof}[Proof of the estimates \eqref{eq:main:ptw:I} in Theorem \ref{thm:regionI}]
We divide the proof into three steps.

\textit{Step 1: Computing the asymptotics of $\Lxi^j\psi$.} By \eqref{energy:decay:redshift-region:zeromode}, we have for the $\ell=0$ mode that
\begin{align}\label{proof:1:234}
  |\Lxi^j\psiz(u,\ub)-\Lxi^j\psiz(-\infty,\ub)|={}&\biggl|\int_{-\infty}^u\pu\Lxi^j\psiz(u',\ub)\di u'\biggr|\notag\\
\leq{}&\biggl(\int_{-\infty}^u-\mu(u',\ub)\di u'\biggr)^\half\biggl(\int_{-\infty}^u(-\mu)^{-1}|\pu\Lxi^j\psiz|^2(u',\ub)\di u'\biggr)^\half\notag\\
={}&\frac{1}{\sqrt{2}} (r_+-r(u,\ub))^\half\biggl(\int_{-\infty}^u(-\mu)^{-1}|\pu\Lxi^j\psiz|^2(u',\ub)\di u'\biggr)^\half\notag\\
\lesssim {}&(-\mu)^\half \ub^{-4-j}.
\end{align}
By  the assumed estimates \eqref{eq:asympp:rred}, one achieves
\begin{align}\label{proof:I:ellz:872}
  |\Lxi^j\psiz(u,\ub)-c_0\Lxi^j(\ub^{-3})|
\lesssim {}&\ub^{-3-j-\epsilon}.
\end{align}

For $\psih$, we integrate
\begin{align} \pu \big(\Lxi^j\psilarge(u,\tub_1+r-r_+,\phiin,\th)\big)=(\pu \Lxi^j\psilarge+\half\mu\pv\Lxi^j\psilarge)(u,\tub_1+r-r_+,\phiin,\th)\end{align}
along $\tub=\tub_1$ from $u=-\infty$ and claim that $|\Lxi^j\psilarge(u,\ub)|\lesssim \ub^{-4-j-\delta}$ in $\I$. In fact,  
\begin{align}
  &\int_{-\infty}^u(\pu \Lxi^j\psilarge+\half\mu\pv\Lxi^j\psilarge)(u',\tub_1+r-r_+,\phiin,\th)\di u'\notag\\
\lesssim &\biggl(\int_{-\infty}^u-\mu(u',\tub_1+r-r_+)\di u'\biggr)^\half\notag\\
&\qquad\times\biggl(\int_{-\infty}^u[(-\mu)^{-1}|\pu\Lxi^j\psih|^2 +(-\mu) |\pv\Lxi^j\psih|^2](u',\tub_1+r-r_+,\phiin,\th)\di u'\biggr)^\half\notag\\
\lesssim &(-\mu)^\half \biggl(\iint\limits_{\{\tub=\tub_1\}\cap\{u'\leq u\}}\big(\Enn[\Lxi^{\leq 1}\Lxi^j(\Carter^{\leq 1}\psi)_{\ell\geq1}]+\Enn[\Lxi^{\leq 1}\Lxi^{j+1}\psi]\big)(-\mu)\dSp\di u'\biggr)^\half\lesssim (-\mu)^\half \tub_1^{-4-j-\delta},
\end{align}
where we have used  $\pv=e_4-\frac{a}{q^2}[\sin\th e_2-a\sin^2\th(e_4+\mu e_3)]$ and the Sobolev imbedding estimate \eqref{SI:estimate:sph} to the second inequality,  and
\eqref{energy:decay:redshift-region:highermode} and \eqref{energy:decay:redshift-region:full} to the last step. The claim then follows from the assumed estimate \eqref{assump:eq:HH:highmodes} on the event horizon. Meanwhile, since $\Carter$ and $\Leta$ commute with the wave operator, we conclude for $k_1\geq 0$, $k_2\geq 0$, and $k_1+k_2\leq 1$,
\begin{align}
\label{eq:ellgeq1:highderi:I}
|\Lxi^j\Phi^{k_1}\Ckpsilarge(u,\ub,\phiin,\th)|\lesssim \ub^{-4-j-\delta},
\end{align}
which together with $\Leta\psi=\Leta\psih$ implies 
\begin{align}\label{pointwise:decay:Leta-scalar:red}
  |\Lxi^j\Leta\psi(u,\ub,\phiin,\th)|\lesssim \ub^{-4-j-\delta}\ \ \text{in}\ \I .
\end{align}

In summary, the estimates \eqref{eq:ellgeq1:highderi:I}, 
 \eqref{proof:1:234} and the assumed estimate \eqref{eq:asympp:rred} together yield
\begin{align}\label{pointwise:decay:scalar:red}
  |\Lxi^j\psi(u,\ub,\phiin,\th)-c_0\Lxi^j(\ub^{-3})|\lesssim \ub^{-3-j-\epsilon}\ \ \text{in}\ \I.
\end{align}
Moreover, we have $|\Carter\psi|\leq |(a^2\sin^2\th\Lxi^2\psi)_{\ell=0}|+|(\Carter\psi)_{\ell\geq 1}|\lesssim \ub^{-4-\delta}$ in $\I$, since $|(a^2\sin^2\th\Lxi^2\psi)_{\ell=0}|\lesssim a^2 \abs{\Lxi^2 \psi}_{L^{\infty}(\Sp)} \lesssim \ub^{-5}$.

\textit{Step 2: Computing the asymptotics of $\pv\psi$.}
 We recast the wave equation \eqref{eq-wave-doublenull-mode:=0} satisfied by $\psiz$ as
\begin{align}
\label{eq:wave:ell=0:step2:I}
\hspace{4ex}&\hspace{-4ex}
\pu\biggl((\R) \pv\psiz-\half(\R)\Lxi\psiz\biggr)\notag\\
  =&-\half(\R)\pu\Lxi\psiz-\frac{1}{4}a^2\mu \PJ_{\ell=0} (\sin^2\theta \Lxi^2\psi)
\end{align}
and integrate this equation along $\ub=constant$ from $u=-\infty$. Note that on $\Horizon$,
\begin{align}
(\R) \pv\psiz-\half(\R)\Lxi\psiz=\half(r_+^2+a^2)\Lxi\psiz=-\frac{3}{2}c_0(r_+^2+a^2)\ub^{-4} + O(\ub^{-4-\epsilon}).
\end{align}
Further, the integral of the first term on the right-hand side of \eqref{eq:wave:ell=0:step2:I} is bounded by $C(-\mu)^\half\ub^{-5}$ by using \eqref{proof:1:234} with $j=1$, and the integral of the second term on the right-hand side of \eqref{eq:wave:ell=0:step2:I} is controlled as follows:
\begin{align}
  &\biggl|\int_{-\infty}^u\mu\PJ_{\ell=0} (\sin^2\theta \Lxi^2\psi)(u',\ub_1)\di u'\biggr|\notag\\
  \lesssim &(-\mu)^\half \biggl(\int_{-\infty}^u-\mu |\PJ_{\ell=0} (\sin^2\theta \Lxi^2\psi)|^2(u',\ub_1)\di u'\biggr)^\half\notag\\
  \lesssim &(-\mu)^\half \biggl(\int_{-\infty}^u\int_{\mathbb{S}_{u',\ub}^2}-\mu |\PJ_{\ell=0} (\sin^2\theta \Lxi^2\psi)|^2(u',\ub_1,\phiin,\th)\dSp\di u'\biggr)^\half\notag\\
  \lesssim &(-\mu)^\half \biggl(\iint\limits_{\{\ub=\ub_1\}\cap\{u'\leq u\}}-\mu |\Lxi^2\psi|^2\dSp\di u'\biggr)^\half\lesssim\ -\mu \ub_1^{-5},
\end{align}
where we have used the inequality \eqref{pointwise:decay:scalar:red} with $j=2$ to the last step.
Therefore, we obtain
\begin{align}
  \bigg|(\R) \pv\psiz(u,\ub)-\half(\R)\Lxi\psiz(u,\ub)+\frac{3}{2}c_0(r_+^2+a^2)\ub^{-4}\bigg|\lesssim \ub^{-4-\epsilon}.
\end{align}
Substituting \eqref{pointwise:decay:scalar:red} back into the above inequality, we arrive at
\begin{align}
\label{eq:precise:pvpsi:ell=0:I}
  \bigg|\pv\psiz(u,\ub)+\frac{3}{2}c_0\bigg(1+\frac{r_+^2+a^2}{\R}\bigg)\ub^{-4}\bigg|\lesssim \ub^{-4-\epsilon}.
\end{align}

We next turn to deriving the pointwise decay estimates for $e_4\psih$.  Instead, we integrate $ \pu \big(e_4\psilarge(u,\tub_1+r-r_+,\phiin,\th)\big)$  along $\tub=\tub_1$ from $u=-\infty$: 
\begin{align}
\label{eq:pvpsielllarge:inte:I}
  \hspace{4ex}&\hspace{-4ex}
  e_4\psilarge(u, \tub_1+r(u,\ub)-r_+,\phiin,\th)-e_4\psilarge(-\infty,\tub_1,\phiin,\th)\notag\\
  =&\int_{-\infty}^u \bigg(\pu e_4\psilarge+\half\mu\pv e_4\psilarge\bigg)(u',\tub_1+r(u',\ub)-r_+,\phiin,\th)\di u'\notag\\
  =&\int_{-\infty}^u\biggl((1+\half\mu)\pu e_4\psilarge+\half\mu e_4\Lxi\psilarge\biggr)(u',\tub_1+r(u',\ub)-r_+,\phiin,\th)\di u'\notag\\
  \lesssim &(-\mu)^\half \biggl(\int_{-\infty}^u\bigl((-\mu)^{-1}|\pu e_4\psilarge|^2+ (-\mu) |e_4\Lxi\psilarge|^2\bigr)(u',\tub_1+r(u',\ub)-r_+,\phiin,\th)\di u'\biggr)^\half\notag\\
  \lesssim &(-\mu)^\half \tub_1^{-4-\delta},
\end{align}
where we have used $\pv=\Lxi+\pu$ to the third line and the following equality
\begin{align}
 & 4(r^2+a^2)\pu e_4(\psilarge)\notag\\
 &\quad =-\mu\biggl(
  2a\Lxi\Leta(\psilarge)
+(\Carter\psi)_{\ell\geq 1}
+2 r(\pu(\psilarge)+ e_4(\psilarge))
+\frac{2ar}{r^2+a^2}\Leta(\psilarge)\biggr),
\end{align}
the Sobolev imbedding \eqref{SI:estimate:sph}  and the estimates
\eqref{energy:decay:redshift-region:highermode} and \eqref{energy:decay:redshift-region:full} to the last line.
Consequently, it holds by the assumed estimate \eqref{assump:eq:HH:highmodes} for $\psilarge$ on $\Horizon$ that
\begin{align}
\label{eq:precise:e4psilarge:I:739}
\abs{e_4\psilarge}\lesssim \ub^{-4-\delta}\ \ \text{in}\ \I.
\end{align} 

 Combining this estimate with the estimate \eqref{eq:precise:pvpsi:ell=0:I} for $\psiz$ and \eqref{pointwise:decay:Leta-scalar:red} which says $|\Leta\psih|\lesssim \ub^{-4-\delta}$, we conclude
\begin{align}
  \label{eq:precise:pvpsi:full:I}
    \bigg|\pv\psi(u,\ub,\phiin,\th)+\frac{3}{2}c_0\bigg(1+\frac{r_+^2+a^2}{\R}\bigg)\ub^{-4}\bigg|\lesssim \ub^{-4-\epsilon}\ \ \text{in}\ \I,
  \end{align}
  and this together with \eqref{pointwise:decay:Leta-scalar:red} gives
  \begin{align}
    \label{eq:precise:e4psi:full:I}
      \bigg|e_4\psi(u,\ub,\phiin,\th)+\frac{3}{2}c_0\bigg(1+\frac{r_+^2+a^2}{\R}\bigg)\ub^{-4}\bigg|\lesssim \ub^{-4-\epsilon}\ \ \text{in}\ \I.
    \end{align}

\textit{Step 3: Computing the asymptotics of $\pu\psi$.} A straightforward application of the estimates \eqref{pointwise:decay:scalar:red} and \eqref{eq:precise:pvpsi:full:I} and the relation $\pu=\pv-T$ yields
 \begin{align}
  \label{eq:precise:pupsi:full:I:1}
    \bigg|\pu\psi(u,\ub,\phiin,\th)-\frac{3}{2}c_0\mu\frac{r+r_+}{r-r_-}\ub^{-4}\bigg|\lesssim \ub^{-4-\epsilon}\ \ \text{in}\ \I.
 \end{align}

This however does not yield the estimate \eqref{eq:main:ptw:I:pupsi} for $e_3\psi$ that we wish to show, thus we shall now improve the estimate \eqref{eq:precise:pupsi:full:I:1} near the event horizon. 
Recall the expression $\partial_r\mu=\frac{2}{\R}(r-M-r\mu)$, thus one can choose $\rred\in (r_-,r_+)$ such that $\partial_r\mu(r)\geq c_0>0$ for $r\in[\rred,r_+]$. It then follows that in $[\rred,r_+)$, 
 \begin{align}\label{proof:1:789}
  e_4\bigl(|\mu|^{-1}|\pu\psi|^2\bigr)
  ={}&\pv((-\mu)^{-1}) |\pu \psi|^2 + 2 \abs{\mu}^{-1} \pu \psi e_4\pu\psi\notag\\
  \leq {}&c\mu^{-1}|\pu\psi|^2+C|\mu|^{-1}|e_4\pu\psi|^2\notag\\
\leq {}&c\mu^{-1}|\pu\psi|^2+C(-\mu)\ub^{-8}
 \end{align}
 for some constants $c>0$ and $C>0$.
Here, we have used $\pv((-\mu)^{-1})=\half\mu^{-1}\partial_r\mu$   and Cauchy-Schwarz in the second line, and $[\pu,e_4]=-\frac{ar\mu}{(\R)^2}\Leta$, equation \eqref{eq-wave-doublenull} and the pointwise decay estimates of $\Carter\psi, \Lxi\Leta\psi, \Leta\psi, \pu\psi, e_4\psi$ to control the term $|\mu|^{-1}|e_4\pu\psi|^2$ in the third line. Integrating \eqref{proof:1:789} in $\{\ub_1<\ub<\ub_2\}\cap \{r_1<r<r_+\}$ for $\rred\leq r_1<r_+$ with respect to $\dSp\di u\di\ub$, we achieve
\begin{align}\label{proof:1:789:1}
  &\iint\limits_{\{\ub=\ub_2\}\cap[r_1,r_+)}|\mu|^{-1}|\pu\psi|^2\dSp\di u
  +\iint\limits_{\{\ub_1\leq\ub\leq\ub_2\}\cap[r_1,r_+)}|\mu|^{-1}|\pu\psi|^2\dSp\di u \di\ub\notag\\
  &{}\ \lesssim \iint\limits_{\{\ub=\ub_1\}\cap[r_1,r_+)}|\mu|^{-1}|\pu\psi|^2\dSp\di u+\iint\limits_{\{\ub=\ub_1\}\cap[r_1,r_+)}(-\mu)\ub^{-8}\di u\di\ub\notag\\
  &{}\ \lesssim \iint\limits_{\{\ub=\ub_1\}\cap[r_1,r_+)}|\mu|^{-1}|\pu\psi|^2\dSp\di u+(r_+-r_1)\int_{\ub_1}^{\ub_2}\ub^{-8}\di\ub,
\end{align}
where we have dropped the boundary term on $\{r=r_1\}$ on the left side. Because of the well-posedness of the scalar wave equation, all the nondegenerate derivatives are finite on $\{\ub=\ub_0\}\cap \DD$ hypersurface for any finite $\ub_0>0$, therefore, it holds $ \iint\limits_{\{\ub=\ub_0\}\cap[r_1,r_+)}|\mu|^{-1}|\pu\Lxi^j\psi|^2\dSp\di u \lesssim (r_+-r_1)$ for any $r_1\in [\rred,r_+)$. Then, we apply Lemma \ref{lemma:mean:decay} with $\alpha=0$ to \eqref{proof:1:789:1} and achieve 
\begin{align}\label{energy:decay:scalar:red}
  \iint\limits_{\{\ub=\ub_1\}\cap[r_1,r_+)}|\mu|^{-1}|\pu\Lxi^j\psi|^2\dSp\di u\lesssim (r_+-r_1)\ub_1^{-8-2j}.
\end{align}

In a similar manner, we get 
\begin{align}\label{energy:decay:phiscalar:red}
  \iint\limits_{\{\ub=\ub_1\}\cap[r_1,r_+)}|\mu|^{-1}(|\pu\Lxi^j\Carter^{\leq 1}\Carter\psi|^2+|\pu\Lxi^j\Carter^{\leq 1}\Leta\psi|^2)\dSp\di u\lesssim (r_+-r_1)\ub_1^{-8-2j-2\delta},
\end{align}
by using  $|\Carter^{k_1}\Leta^{k_2}X\Lxi^j\psi|\lesssim \ub^{-4-j-\delta}$ with $X\in \{\Carter,\Leta,\pu,\pv\}$ and $k_1+k_2=1$. 

An application of \eqref{energy:decay:phiscalar:red} yields
\begin{align}\label{pointwise:decay:scalar:red:1}
  &|\Lxi^j\psi(u,\ub_1,\phiin,\th)-\Lxi^j\psi(-\infty,\ub_1,\phiin,\th)|\notag\\={}&\biggl|\int_{-\infty}^u\pu\Lxi^j\psi(u',\ub_1,\phiin,\th)\di u'\biggr|\notag\\
\leq{}&\biggl(\int_{-\infty}^u-\mu(u',\ub_1)\di u'\biggr)^\half\biggl(\int_{-\infty}^u(-\mu)^{-1}|\pu\Lxi^j\psi|^2(u',\ub_1,\phiin,\th)\di u'\biggr)^\half\notag\\
\lesssim{}&\frac{1}{\sqrt{2}} (r_+-r(u,\ub_1))^\half\biggl(\iint\limits_{\{\ub=\ub_1\}\cap[r(u,\ub_1),r_+)}(-\mu)^{-1}(|\pu \Carter^{\leq 1}\Lxi^j\psi|^2+|\pu\Lxi^{j+2}\psi|^2)\dSp\di u\biggr)^\half\notag\\
\lesssim {}&(-\mu)\ub_1^{-4-j},
\end{align}
where we have used \eqref{proof:sb:im:234} to the second last line.
Similarly, we also have
\begin{align}\label{pointwise:decay:phiscalar:red}
  |\Leta\psi(u,\ub,\phiin,\th)-\Leta\psi(-\infty,\ub,\phiin,\th)|\lesssim (-\mu)\ub^{-4-\delta}.
\end{align}

Integrating the following formula
 \begin{align}
  4\pu\bigg((\R)\bigg(e_4\psi-\half\Lxi\psi\bigg)\bigg)
  =-2(\R)\pu\Lxi\psi-\mu\Carter\psi-2\mu a\Lxi\Leta\psi
\end{align}
 along $\ub=constant$,
we arrive at
\begin{align}
  \biggl|\bigg(e_4\psi-\half \Lxi\psi\bigg)(u,\ub,\phiin,\th)-\frac{r_+^2+a^2}{\R}\cdot\bigg(e_4\psi-\half\Lxi\psi\bigg)(-\infty,\ub,\phiin,\th)\biggr|\lesssim (-\mu)\ub^{-4-\delta}.
\end{align}
Substituting \eqref{pointwise:decay:phiscalar:red} into this inequality and using $\pv\psi=\Lxi\psi$ on $\Horizon$, this yields 
\begin{align}
  \biggl|(\pv\psi+\pu\psi)(u,\ub,\phiin,\th)-\frac{r_+^2+a^2}{\R}\cdot\pv\psi(-\infty,\ub,\phiin,\th)\biggr|\lesssim (-\mu)\ub^{-4-\delta}.
\end{align}
We use \eqref{pointwise:decay:scalar:red:1} again and achieve
\begin{align}
  \biggl|\pu\psi(u,\ub,\phiin,\th)+\half \mu\frac{r+r_+}{r-r_-}\cdot \pv\psi(-\infty,\ub,\phiin,\th)\biggr|\lesssim (-\mu)\ub^{-4-\delta},
\end{align}
which is, by \eqref{eq:precise:pvpsi:full:I} and $e_3=-\mu^{-1} \pu$, simplified to
\begin{align}
  \biggl|e_3\psi(u,\ub,\phiin,\th)+\frac{3}{2}c_0\frac{r+r_+}{r-r_-}\ub^{-4}\biggr|\lesssim \ub^{-4-\epsilon}.
\end{align}
This completes the proof.
\end{proof}

\section{Proof of Theorem \ref{thm:main}}
\label{sect:II}

In this section, we provide a proof of our main Theorem \ref{thm:main}. In Sections \ref{sect:EnerDec:II} and \ref{sect:precise:II}, we derive the energy decay estimates and prove the precise late-time asymptotic estimates in the region $\II\doteq\DD\cap\{\ub\geq 1\}\cap\{r_-\leq r\leq \rblue\}$ (close to the right Cauchy horizon $\CHorizon$) and the region $\II'\doteq \DD\cap\{u\geq 1\}\cap\{r_-\leq r\leq \rblue\}$ (close to the left Cauchy horizon $\CHorizon'$), under the assumed estimates \eqref{ass:thm:horizon:psi} on the event horizon. In the last Section \ref{sect:pf:mainthm}, we complete the proof of Theorem \ref{thm:main}.

We divide the region $\II$ into two subregions by a hupersursurface $\Gamma$, where the  hypersurface $\Gamma$  and the subregion $\IIf$ in $\II$ are defined  by
\begin{align}\label{def:Gamma:curve}
 \Gamma\doteq \II\cap\{2r^*(u,\ub)= \ub^{\gamma}\},\qquad
 \IIf \doteq \II\cap\{2\rblue^*\leq 2r^*(u,\ub)\leq \ub^{\gamma }\},
\end{align}
with $0<\gamma<1$ being a suitably small constant to be fixed. The other subregion is $\II\backslash \IIf$. Similarly, we define a hypersurface $\Gamma'\doteq \II'\cap \{2r^* = u^{\gamma}\}$ and the subregions $\IIfp$ and $\II'\backslash\IIfp$ of region $\II'$. See Figure \ref{fig:7} for each divided subregions of region $\DDr$.

\begin{figure}[htbp]
  \begin{center}
\begin{tikzpicture}[scale=1.2]
  \fill[blue!40] (1.95,2.05) arc(45:60:3.75) arc(60:126:1.35) arc(92:98.5:2.45 and 2.5) arc(219.5:217.94:23.1)--(0,4)--cycle;
  \fill[green!20] (1.95,2.05) arc(45:60:3.75) arc(60:125:1.35) arc(92:45:3 and 2.03);
  \fill[red!60] (1.95,2) arc(45:102:2.75 and 2) arc(219.52:225:23.1)--(1.95,1.95);
  \draw[] (0.05,0.05) -- (1.95,1.95);
  \draw[] (1.95,2.05) -- (0.05,3.95);
  \draw[] (-1.95,1.95) -- (-0.05,0.05);
  \draw[] (-1.95,2.05) -- (-0.05,3.95);
  \draw[] (1.95,2) arc(45:135:2.75 and 2);
  \draw[] (1.95,2.05) arc(45:60:3.75) arc(60:135:1.35);
  \draw[] (-0.96,3.04) arc(218:225:23.1);
  \draw[] (2,2) circle (0.05);
  \draw[] (-2,2) circle (0.05);
  \draw[fill=black] (0,0) circle (0.05);
  \draw[fill=black] (0,4) circle (0.05);
  \node at (0.2,-0.3) {$\Sphere_{\HH}$};
  \node at (0.2,4.3) {$\Sphere_{\CHH}$};
  \node at (1.2,3.3) {$\CHorizon$};
  \node at (-1.2,3.3) {$\CHorizon'$};
  \node at (2.2,2.2) {$i_+$};
  \node at (-2.2,2.2) {$i_+'$};
  \node at (0.55,2.68) {$\IIf$};
  \node at (0.8,1.8) {$\I$};
  \node at (0.08,3.15) {$\II\backslash \IIf$};
  \node[rotate=20] at (-1.1,2.25) {$r=\rblue$};
  \node[rotate=-51] at (0.15,1.5) {$\ub= 1$};
  \draw[->] (1.8,3)--(0.95,2.75);
  \node at (1.9,3) {$\Gamma$};
\end{tikzpicture}
\end{center}
\caption{Subregions $\I$, $\IIf$ and $\II\backslash\IIf$ in $\DDr$}
\label{fig:7}
\end{figure}

For the same reason as stated in Section \ref{sect:stability}, we consider only the region $\II$ in Sections \ref{sect:EnerDec:II} and \ref{sect:precise:II}, and only the region $\DD\cap\{\ub\geq 1\}$ in Section \ref{sect:pf:mainthm}. \textbf{Recall from Section \ref{subsect:almostnull} that we use the coodintates $(u,\ub,\phiout,\th)$ in this region $\II$ and  the derivatives in this coordinates are expressed in formula \eqref{eq:deris:uub:nullframe:II}.}

\subsection{Energy decay estimates in region $\II$}
\label{sect:EnerDec:II}

We first show the energy decay estimates for $\psi$.   

Multiplying the scalar wave equation \eqref{eq:scalar:form:e3e4} by $\mu$ on both sides, it becomes
\begin{align}\label{eq-wave-doublenull-2}
  4(r^2+a^2)\Y \pv\psi+\mu\Delta_{\Sphere}\psi
  +2a\mu\Lxi\Leta\psi+\mu a^2 \sin^2 \theta\Lxi^2\psi
  +2\mu r(\Y\psi+ \pv\psi)
  +\frac{2ar\mu}{r^2+a^2}\Leta\psi=0,
\end{align}
where we have defined $\Y\doteq -\mu e_3=\pu-\frac{a}{r^2+a^2}\Leta$ for convenience.

\begin{prop}
  Under the assumed estimates \eqref{ass:thm:horizon:psi} in Theorem \ref{thm:main}, we have, for any $\tub_1\geq1$ such that  $(\tu_1,\tub_1)$ lies in region $\II$, 
  \begin{align}\label{energy:decacy:full:blue:v}
    &\iint\limits_{\{\tu\leq\tu_1, \tub=\tub_1\}\cap\II}\Enw^{(-\half)}[\Lxi^j\psi]\dSp\di u
    +\iint\limits_{\{\tu=\tu_1,\ \tub\geq \tub_{\rblue}(\tu_1)\}}\Ende^{(-\half)}[\Lxi^j\psi]\dSp\di \ub\notag\\
  &+\iint\limits_{\{\tu\leq\tu_1\}\cap\{r_-<r<\rblue\}}\Endeh^{(-\frac{3}{2})}[\Lxi^j\psi]\dSp\di u\di\ub 
    \lesssim \ (\max\{\tub_{\rblue}(\tu_1), 1\})^{-7-2j},
  \end{align}
  and
   \begin{align}\label{energy:decacy:full:blue:w}
  \iint\limits_{\{ \tub=\tub_1\}\cap\IIf}\Enw^{(-\half )}[\Lxi^j\psi]\dSp\di u &\lesssim \tub_1^{-8-2j+\gamma},
   \end{align}
   where $\tub_{\rblue}(\tu_1)$ is such that $r(\tu_1, \tub_{\rblue}(\tu_1))=\rblue$ and the energies are
   \begin{subequations}
   \begin{align}
    \Ende^{(\alpha)}[\psi]&=|\pv\psi|^2-\mu\lmu^{\alpha}\bigl(|\Y\psi|^2+(e_1(\psi))^2+(e_2(\psi))^2\bigr),\\
    \Enw^{(\alpha)}[\psi]&=\lmu^{\alpha}|\Y\psi|^2-\mu\big(|\pv\psi|^2+(e_1(\psi))^2+(e_2(\psi))^2\big),\\
    \Endeh^{(\alpha)}[\psi]&=\lmu^{\alpha}|\Y\psi|^2+\lmu^{-\frac{3}{2}}|\pv\psi|^2-\mu\big((e_1(\psi))^2+(e_2(\psi))^2\big).
   \end{align}
   \end{subequations}

   Moreover, we have for any $\tub_1\geq 1$ that 
   \begin{align}\label{energy:bound:blue:IniEnergyplus}
    &\iint\limits_{\{ \tub=\tub_1\}\cap\II}\Enw^{(-\half)}[\Lxi^j\psi]\dSp\di u
  +\iint\limits_{\{r_-<r<\rblue\}}\Endeh^{(-\frac{3}{2})}[\Lxi^j\psi]\dSp\di u\di\ub \lesssim 1,
   \end{align}
 \end{prop}
 
 \begin{proof}
 Similar to the discussions in Section \ref{sect:ED:scalar:I}, we multiply equation \eqref{eq-wave-doublenull-2} with $X(\psi)\doteq f(r)\Y\psi+g(r)\pv\psi$, integrate on $S_{u,\ub}$ and get
 \begin{align}\label{eq:div:fullfield-2}
  \pv\biggl(\int_{\Sp}\fv[\psi]\dSp\biggr)+\pu\biggl(\int_{\Sp}\fu[\psi]\dSp\biggr) +\int_{\Sp}\mathcal{B}[\psi]\dSp =0,
\end{align}
where
\begin{subequations}
\begin{align}\label{energy:flux:boundary-2}
  \fv[\psi]&=2(r^2+a^2)f(r)|\Y\psi|^2-\half \mu g(r) \big[(e_1(\psi))^2+(e_2(\psi))^2\big]+a\mu \sin\theta e_2(\psi) X(\psi),\\
  \fu[\psi]&=2(r^2+a^2)g(r)|\pv\psi|^2-\half \mu f(r) \big[(e_1(\psi))^2+(e_2(\psi))^2\big]-a\mu \sin\theta e_2(\psi) X(\psi),
\end{align}
and 
\begin{align}\label{energy:spacetime-term-2}
  \mathcal{B}[\psi] =&\bigl(2r\mu g(r)-2\Y((r^2+a^2)g(r))\bigr)|\pv\psi|^2+\bigl(2r\mu f(r)-2\pv((r^2+a^2)f(r))\bigr)|\Y\psi|^2\notag\\
  &+\half[\Y(\mu f(r))+\pv(\mu g(r))]\big((e_1(\psi))^2+(e_2(\psi))^2\big)\notag\\
  &+2\mu r(g(r)+f(r))\pv\psi \Y\psi
  +\frac{2ar\mu}{r^2+a^2}\Leta\psi(-f(r
  )\Y\psi+g(r)\pv\psi).
\end{align}
\end{subequations}

 Taking $f(r)=(\lmu)^{-s}$ and $g(r)=(\R)^p$ with $p>0,s>0$ to be fixed,  the coefficients in $\mathcal{B}[\psi]$ can be estimated as follows:
 \begin{subequations}\label{pf:2:11}
 \begin{align}
   r\mu g(r)-\Y((r^2+a^2)g(r))=&-\mu pr(r^2+a^2)^p,\\
   r\mu f(r)-\pv((r^2+a^2)f(r))=&
   - s(r-M-r\mu)\lmu^{-s-1}\geq C_1 s\lmu^{-s-1},\\
  \Y(\mu f(r))+\pv(\mu g(r))=&
   (\R)^{-1}(r-M-r\mu)\mu(1+s\lmu^{-1})\lmu^{-s}\notag\\
  &+\mu[(r-M)+(p-1)r\mu](\R)^{p-1}\notag\\
  \geq &-C_1(\R)^{-1}\mu \lmu^{-s}-C_1\mu (\R)^{p-1},
 \end{align}
\end{subequations}
 for $r\in (r_-,r_1]$, where $r_1$ is given to satisfy  $\min\limits_{r\in[r_-,r_1]}\{-r+M+(1-p)r\mu\}= C_1>0$ and we have used $\partial_r\mu=2(\R)^{-1}(r-M-r\mu)$, $\pv\mu=\half\mu\partial_r\mu$ and $\pu\mu=\half\mu\partial_r\mu$.

 For any fixed $s>0$,  we deduce that there are constants  $p>0$ and $\rblue\in (r_-,r_1]$, both of which depend on $s$, such that 
 \begin{align}\label{estimate:spacetimeterm:2}
  \mathcal{B}[\psi] &\gtrsim_{p,s,\rblue} \lmu^{-s-1}|\Y\psi|^2+(-\mu)\big(|\pv\psi|^2+(e_1(\psi))^2+(e_2(\psi))^2\big), \ \ r\in(r_-,\rblue].
 \end{align}
 
 In fact, we apply Cauchy--Schwarz inequality to the last line in \eqref{energy:spacetime-term-2} and use $
  \frac{q^2}{\R} \Leta\psi
   =\sin\th e_2(\psi)-a\sin^2\th (\pv\psi -\Y\psi)$, then we get
   \begin{subequations}\label{pf:2:111}
\begin{align}
  2\mu r(g(r)+f(r))\pv\psi \Y\psi \geq& \mu r(\lmu^{-s}+(\R)^p)\bigg(\epsilon_1 |\Y\psi|^2+\frac{1}{\epsilon_1}|\pv\psi|^2\bigg)\notag\\
  =&\mu r\biggl(\lmu +(\R)^p\lmu^{s+1}\biggr)\epsilon_1(\lmu)^{-s-1}|\Y\psi|^2\notag\\
  &+\biggl(\lmu^{-s}(\R)^{-p}+1\biggr)\frac{1}{\epsilon_1}\mu r(\R)^p|\pv\psi|^2\notag\\
  \geq & -2\epsilon_1\lmu^{-s-1}|\Y\psi|^2+\frac{2}{\epsilon_1}\mu r(\R)^p|\pv\psi|^2,\\
  -\frac{2ar\mu f(r)}{r^2+a^2}\Leta\psi \Y\psi
  =&-\frac{2r\mu f(r)}{q^2}\bigg(a^2\sin^2\th |\Y\psi|^2+a\sin\th e_2(\psi)\Y\psi-a^2\sin^2\th \pv\psi \Y\psi \biggr) \notag\\
  \geq &\frac{2r\mu \lmu^{-s}}{q^2}\big(\epsilon_2^{-1}|\Y\psi|^2 +\big(a^2|\pv\psi|^2+\epsilon_2(e_2(\psi))^2\big) \big)\notag\\
  \geq &\frac{2\mu \lmu^{-s}}{\epsilon_2 r}|\Y\psi|^2+\mu\big(a^2r(\R)^p|\pv\psi|^2+\epsilon_2(\R)^{p-1}(e_2(\psi))^2\big),\\ 
 \frac{2ar\mu g(r)}{r^2+a^2}\Leta\psi\pv\psi 
\geq & \frac{a^2r\mu g(r)}{q^2}\frac{2}{\epsilon_3}|\pv\psi|^2 +\frac{r(\R)\epsilon_3}{q^2}\mu (\R)^{p-1}(e_2(\psi))^2\notag\\
& -\epsilon_3\lmu^{-s-1}|\Y\psi|^2,
 \end{align}
\end{subequations}
 for $r\in (r_-,r_2]$, where $r_2$ is chosen depending on $p,s$ such that all the functions $-\mu r\lmu$, $-\mu r(\R)^p\lmu^{s+1}$, $(\R)^{-p}\lmu^{-s}$, $r^{-2}(\R)^{-p}\lmu^{-s}$, $-r^{-1}\mu \lmu^{s+1}$ and $r^{-1}(\R)^{-p+1}\lmu^{-s}$ are  bounded by $1$ in $(r_-,r_2]$. 
 We choose $\epsilon_1=\half C_1 s$, $\epsilon_2=\frac{C_1}{4}$ and $\epsilon_3=\min\{\frac{C_1 r_-^2}{8|a|r_+(r_+^2+a^2)},\frac{C_1 s}{2|a|}\}$, then the terms in \eqref{pf:2:111} are absorbed by the terms in \eqref{pf:2:11} if taking $p$ sufficiently large and choosing $\rblue\in (r_-,\min\{r_1,r_2\}]$ such that $\frac{-2\mu\lmu}{\epsilon_2 r}\leq \frac{1}{4}C_1 s$ for $r\in(r_-,\rblue]$. This yields \eqref{estimate:spacetimeterm:2}.

For $(\tu_1, \tub_1)\in \II$, integrating  \eqref{eq:div:fullfield-2} in $\{(u,\ub)| \tu\leq \tu_1,\ \tub\leq \tub_1\}\cap \{r_-<r<\rblue\}$ with respect to $\di u \di\ub$, we have
\begin{align}\label{energy:identity:blue:1}
  &\iint\limits_{\{\tu_{\rblue}(\tub_1)\leq\tu\leq\tu_1, \tub=\tub_1\}}\mathbb{T}_{\tub}[\psi]\dSp\di u+\iint\limits_{\{\tu=\tu_1, \tub_{\rblue}(\tu_1)\leq \tub\leq \tub_1\}}\mathbb{T}_{\tu}[\psi]\dSp\di \ub\notag\\
  &+\iint\limits_{\{\tu\leq\tu_1,  \tub\leq \tub_1\}\cap\{r_-<r<\rblue\}}\mathcal{B}[\psi]\dSp\di u\di\ub =\iint\limits_{\{r=\rblue\}\cap\{\tub_{\rblue}(\tu_1)\leq \tub\leq \tub_1\}}\mathbb{T}_{r}[\psi]\dSp\di\ub,
\end{align}
where $\Entub=(1-\half\mu)\fv[\psi]-\half\mu\fu[\psi]$, $\Enr=-\half\mu (\fv[\psi]+\fu[\psi])$ and $\mathbb{T}_{\tu}[\psi]=(1-\half\mu)\fu[\psi]-\half\mu\fv[\psi]$. 
Now we consider the boundary terms. Substituting the following inequalities
\begin{subequations}
\begin{align}
 |a\mu \sin\theta e_2(\psi) f(r)\Y\psi|&\leq a^2 f(r)|\Y\psi|^2+\frac{1}{4}\mu^2 f(r)(e_2(\psi))^2\\
 |a\mu \sin\theta e_2(\psi) g(r)\pv\psi|&\leq -\frac{3}{4}a^2\mu g(r)|\pv\psi|^2-\frac{1}{3}\mu g(r)(e_2(\psi))^2
\end{align}
\end{subequations}
into $\Entub$, we obtain
\begin{align}
 \Entub &\thicksim \lmu^{-s}|\Y\psi|^2-\mu\bigl(|\pv\psi|^2+(e_1(\psi))^2+(e_2(\psi))^2\bigr).
\end{align}
Similarly, we also have
\begin{align}
 \mathbb{T}_{\tu}[\psi] &\thicksim 
  |\pv\psi|^2-\mu\lmu^{-s}\bigl(|\Y\psi|^2+(e_1(\psi))^2+(e_2(\psi))^2\bigr).
\end{align}
We fix $s=\half$ in the above equality \eqref{energy:identity:blue:1}. Consequently, we get
\begin{align}\label{energy:decacy:full:blue:v:11}
  &\iint\limits_{\{\tu\leq\tu_1, \tub=\tub_1\}\cap\II}\Enw^{(-\half)}[\Lxi^j\psi]\dSp\di u+\iint\limits_{\{\tu=\tu_1,\ \tub\geq \tub_{\rblue}(\tu_1)\}}\Ende^{(-\half)}[\Lxi^j\psi]\dSp\di \ub\notag\\
&\quad +\iint\limits_{\{\tu\leq\tu_1\}\cap\{r_-<r<\rblue\}}\Enw^{(-\frac{3}{2})}[\Lxi^j\psi]\dSp\di u\di\ub 
  \lesssim \ (\max\{\tub_{\rblue}(\tu_1),1\})^{-7-2j},
\end{align}
in which we have used \eqref{energy:decay:red:spacetime-esti} for $\tub_{\rblue}(\tu_1)\geq 1$ or Remark \ref{remark:boundedness:rblue} for $\tub_{\rblue}(\tu_1)<1$ to control the boundary terms on $r=\rblue$. To show \eqref{energy:decacy:full:blue:v}, it suffices to improve the spacetime estimate of $\pv\psi$. In fact, one can take $f(r)=\lmu^{-\half}$ and $g(r)=\lmu^{-\half}$ in \eqref{energy:identity:blue:1} and run again the above argument, then \eqref{energy:decacy:full:blue:v} follows. Moreover,  taking $\tu_1\to+\infty$ (and thus $\tub_{\rblue}(\tu_1)\to -\infty$) in \eqref{energy:decacy:full:blue:v} and by Remark \ref{remark:boundedness:rblue} on the boundedness of the energy on ${r=\rblue}$, we get the estimate \eqref{energy:bound:blue:IniEnergyplus}.

To prove \eqref{energy:decacy:full:blue:w}, we integrate  \eqref{eq:div:fullfield-2} in $\{(u,\ub)| \tub_1\leq \tub\leq \tub_2\}\cap \IIf$ with respect to $\di u \di\ub$ and obtain
\begin{align}\label{energy:identity:blue:IIf}
  &\iint\limits_{\{ \tub=\tub_2\}\cap\IIf}\mathbb{T}_{\tub}[\psi]\dSp\di u+\iint\limits_{\{\tub_1\leq \tub\leq \tub_2\}\cap\Gamma}\mathbb{T}_{\Gamma}[\psi]\dSp\di \ub
  +\iint\limits_{\{\tub_1\leq \tub\leq \tub_2\}\cap\IIf}\mathcal{B}[\psi]\dSp\di u\di\ub \notag\\
  =&\iint\limits_{\{ \tub=\tub_1\}\cap\IIf}\mathbb{T}_{\tub}[\psi]\dSp\di u
  +\iint\limits_{\{r=\rblue\}\cap\{\tub_1\leq \tub\leq \tub_2\}}\mathbb{T}_{r}[\psi]\dSp\di\ub,
\end{align}
where $\Entub=(1-\half\mu)\fv[\psi]-\half\mu\fu[\psi]$, $\Enr=-\half\mu (\fv[\psi]+\fu[\psi])$ and $\mathbb{T}_{\Gamma}[\psi]=\fu[\psi]+(1-\gamma \ub^{\gamma-1})\fv[\psi]$. We note that $\mathbb{T}_{\Gamma}[\psi]\gtrsim f(r)|\Y\psi|^2 +g(r)|\pv\psi|^2- \mu (g(r)+f(r)) \big[(e_1(\psi))^2+(e_2(\psi))^2\big]$ for $\gamma>0$ sufficiently small.

The same argument as in proving \eqref{energy:decacy:full:blue:v} applies and, using \eqref{energy:decay:red:spacetime-esti} to control the boundary term on $\{r=\rblue\}$, we achieve, for any $s>0$, 
 \begin{align}\label{proof:energy:inequality}
  &\iint\limits_{\{ \tub=\tub_2\}\cap\IIf}\Enw^{(-s)}[\Lxi^j\psi]\dSp\di u+\iint\limits_{\{ \tub_1\leq \tub\leq \tub_2\}\cap\IIf}\Enw^{(-s-1)}[\Lxi^j\psi]\dSp\di u\di\ub \notag\\
  \lesssim &\iint\limits_{\{ \tub=\tub_1\}\cap\IIf}\Enw^{(-s)}[\Lxi^j\psi]\dSp\di u+\tub_1^{-8-2j}+\int_{\tub_1}^{\tub_2}\tub^{-8-2j}\di \tub.
\end{align}
By definition, one finds 
\begin{align}\Enw^{(-s-1)}[\Lxi^j\psi]\geq \lmu^{-1}\Enw^{(-s)}[\Lxi^j\psi]\sim (r^*)^{-1}\Enw^{(-s)}[\Lxi^j\psi]\geq \ub^{-\gamma}\Enw^{(-s)}[\Lxi^j\psi]\sim \tub^{-\gamma}\Enw^{(-s)}[\Lxi^j\psi]
\end{align} in $\IIf$. Hence, 
by applying  Lemma \ref{lemma:mean:decay} with $\alpha=\gamma$ to \eqref{proof:energy:inequality} with $s=\half$, we obtain
\eqref{energy:decacy:full:blue:w}. 
 \end{proof}

 Next, we derive  the energy decay estimates for both $\psi_{\ell=0}$ and $\psi_{\ell\geq 1}$  in  region $\II$.
\begin{prop}
  Under the assumed estimates \eqref{ass:thm:horizon:psi} in Theorem \ref{thm:main}, we have for any $(u_1,\ub_1)\in \II$,
  \begin{align}\label{energy:decay:blue:zeromode}
    &\int_{\ub_{\rblue}(u_1)}^{+\infty}|\pv\Lxi^j\psiz|^2(u_1,\ub)\di \ub 
    +\int_{u_{\rblue}(\ub_1)}^{u_1}\lmu^{-\half}|\pu\Lxi^j\psiz|^2(u,\ub_1)\di u \notag\\
   &+\int_{\ub_{\rblue}(u_1)}^{+\infty}\int_{u_{\rblue}(\ub)}^{u_1}\Big(\lmu^{-\frac{3}{2}}|\pu\Lxi^j\psiz|^2+\lmu^{-\frac{3}{2}}|\pv\Lxi^j\psiz|^2\Big)(u,\ub)\di u\di \ub \notag\\
  & \lesssim (\max\{\ub_{\rblue}(u_1),1\})^{-7-2j},
  \end{align}
  and, for any $\tub_1\geq 1$, $k_1\geq 0$ and $k_2\geq 0$,
  \begin{subequations}
  \begin{align}\label{energy:decay:blue:higher:ub}
    \iint\limits_{\{ \tub=\tub_1\}\cap\IIf} \Enw^{(-\half)}[\Lxi^j\Leta^{k_1}\Ckpsilarge] \dSp\di u &\lesssim \tub_1^{-8-2j-2\delta+\gamma},\\
    \iint\limits_{\{ \tub=\tub_1\}\cap\II} \Enw^{(-\half)}[\Lxi^j\Leta^{k_1}\Ckpsilarge] \dSp\di u &\lesssim 1.\label{energy:bound:blue:higher:ub}
  \end{align}
\end{subequations}
\end{prop}
\begin{proof}
  The proof is the same as the one for Lemma \ref{Lemma:mode:red}, and we omit it.
\end{proof}

Finally, we derive the boundedness estimate for $\Lxi^j\psi$ by some Hardy-type argument.

\begin{cor}
  Under the assumed estimates \eqref{ass:thm:horizon:psi} in Theorem \ref{thm:main}, we have 
 \begin{subequations}
    \begin{align}\label{energy:decay:blue:psiz:nearHor}
       \int_{u_{\rblue}(\ub_1)}^{+\infty} \lmu^{-\frac{5}{2}}|\Lxi^j\psiz|^2(u,\ub_1)\di u \lesssim {}&1,\\
\label{energy:decay:blue:psi:nearHor}
    \iint\limits_{\{\ub=\ub_1\}\cap\II}\lmu^{-\frac{5}{2}}|\Lxi^j\psi|^2\dSp\di u\lesssim {}&1.
  \end{align}
  \end{subequations}
  \end{cor}
  
  \begin{proof}
    A direct calculation yields
    \begin{align}
      \pv\big(\lmu^{-p}(\Lxi^j\psiz)^2\big)={}& \half p\partial_r\mu\lmu^{-p-1}|\Lxi^j\psiz|^2+2\lmu^{-p}\Lxi^j\psiz \pv\Lxi^j\psiz\notag\\
      \leq {}&\half p\partial_r\mu\lmu^{-p-1}|\Lxi^j\psiz|^2
      +\delta \lmu^{-p-1}|\Lxi^j\psiz|^2\notag\\
      &+\frac{1}{\delta}\lmu^{-p+1}|\pv\Lxi^j\psiz|^2 .
    \end{align}
 We integrate the above inequality in $\{(u,\ub)|u\leq u_1, \ub\leq\ub_1\}\cap \{r_-<r<\rblue\}$,  take $\delta$ sufficiently small and obtain for $p>1$ that 
    \begin{align}
      &\int_{u_{\rblue}(\ub_1)}^{u_1}\lmu^{-p}\abs{\Lxi^j\psiz}^2(u,\ub_1)\di u 
      +\int_{\ub_{\rblue}(u_1)}^{\ub_1}\int_{u_{\rblue}(\ub)}^{u_1}\lmu^{-p-1}\abs{\Lxi^j\psiz}^2(u,\ub)\di u \di \ub\notag\\
      \lesssim &{}\int_{\ub_{\rblue}(u_1)}^{\ub_1}\int_{u_{\rblue}(\ub)}^{u_1}\lmu^{-p+1}\abs{\pv\Lxi^j\psiz}^2(u,\ub)\di u \di \ub+\int_{\ub_{\rblue}(u_1)}^{\ub_1}\lmu^{-p}\abs{\Lxi^j\psiz}^2(2\rblue^*-\ub,\ub)\di \ub.
    \end{align}
By \eqref{energy:decay:blue:zeromode}, the first term in the second line is bounded by $(\max\{\ub_{\rblue}(u_1),1\})^{-7-2j}$  for $p=\frac{5}{2}$. Then \eqref{energy:decay:blue:psiz:nearHor}  holds by the estimates on $r=\rblue$.

    The proof of \eqref{energy:decay:blue:psi:nearHor} is similar, so we omit it.
  \end{proof}

\subsection{Precise late-time asymptotics in region $\II$} 
\label{sect:precise:II}

We derive the precise late-time asymptotics for the scalar field and its derivatives in the subregion $\IIf$ of $\II$ and in the entire $\II$ in Sections \ref{sect:IIf} and \ref{sect:IImoduloIIf}, respectively.


\subsubsection{Estimates in $\IIf$}
\label{sect:IIf}

  {\it Step 1. Computing the asymptotics of $\Lxi^j\psi$ in $\IIf$.}  For any $(u,\ub)\in \IIf$ and $\ub\gg 1$, we have
  \begin{align}\label{proof:2:iif}
    0\leq \ub-\ub_{\rblue}(u)\leq \ub^{\gamma}-2\rblue^*,\quad 0\leq u-u_{\rblue}(\ub)\leq \ub^{\gamma}-2\rblue^* , \ \ 1+\frac{2\rblue^*}{\ub}-\ub^{\gamma-1}\leq \frac{\ub_{\rblue}(u)}{\ub}\leq 1.
  \end{align}
  In fact, the definition of $\IIf$ implies $2\rblue^*=\ub_{\rblue}(u)+u \leq \ub_{\rblue}(u)-\ub+\ub^{\gamma}$ and $2\rblue^*=u_{\rblue}(\ub)+\ub \leq u_{\rblue}(\ub)-u+\ub^{\gamma}$, then the above inequalities are straightforward.

  Integrating $\pv\Lxi^j\psiz$ along $\ub=constant$ from $\{r=\rblue\}$, we derive
  \begin{align}
    |\Lxi^j\psiz(u,\ub)-\Lxi^j\psiz(u,\ub_{\rblue}(u))|&=\bigg|\int_{\ub_{\rblue}(u)}^{\ub} \pv\Lxi^j\psiz(u,\ub')\di \ub'\bigg|\notag\\
    &\lesssim (\ub-\ub_{\rblue}(u))^\half
    \biggl(\int_{\ub_{\rblue}(u)}^{\ub} |\pv\Lxi^j\psiz|^2(u,\ub')\di \ub'\biggr)^\half\notag\\
    &\lesssim (\ub^{\gamma}-2\rblue^*)^\half (\ub_{\rblue}(u))^{-\frac{7}{2}-j}\notag\\
    &\lesssim (\ub_{\rblue}(u))^{-\frac{7}{2}-j+\half \gamma}\lesssim \ub^{-\frac{7}{2}-j+\half \gamma},
  \end{align}
  where we have used \eqref{energy:decay:blue:zeromode} in the third step.
  This estimate, the estimate \eqref{proof:1:234}, the second inequality in \eqref{proof:2:iif} and the assumption on event horizon together yield
  \begin{align}\label{pointwise:decay:psiz:blue:away}
    |\Lxi^j\psiz(u,\ub)-c_0\Lxi^j(\ub^{-3}))|\lesssim \ub^{-3-j-\epsilon}\ \ \text{for}\ (u,\ub)\in\IIf.
  \end{align}

  In coordinates $(u,\ub,\phiout,\th)$, for any $\tub_1$ and $(\phiin)_1$, we define a curve $\gamma_{\tub_1,(\phiin)_1}(u): u\mapsto (u,\ub(u), \phiout(u),\th(u))$ with parameter $u$, satisfying that $\ub(u)-r(u,\ub(u))+r_+=\tub_1$, $\phiout(u)+2\rmod(u,\ub(u))=(\phiin)_1$ and $\th(u)=const$.

  For $\Lxi^i\psih$, we integrate $\pu(\Lxi^j\psih\big|_{\Cur})=(\Y\Lxi^j\psih+\half\mu\pv\Lxi^j\psih)\big|_{\Cur}$ along $\Cur$ from $\{r=\rblue\}$ and obtain
  \begin{align}
    &|\Lxi^j\psih(u,\tub_1+r,(\phiin)_1-2\rmod,\th)-\Lxi^j\psih(u_{\rblue}(\tub_1),\tub_1+r,(\phiin)_1-2\rmod,\th)|\notag\\
    \lesssim {}&\biggl(\int_{u_{\rblue}(\tub_1)}^u\lmu^{\half}(u',\tub_1+r)\di u'\biggr)^\half\notag\\
    &\quad \times \biggl(\int_{u_{\rblue}(\tub_1)}^{u} \Big(\lmu^{-\half}\big(|\Y\Lxi^j\psih|^2+\mu^2|\pv\Lxi^j\psih|^2\big)\Big)\Big|_{\Cur(u')}\di u'\biggr)^\half\notag\\
    \lesssim {}&\tub_1^{\frac{3}{4}\gamma}
    \biggl( \iint\limits_{\tub=\tub_1, u_{\rblue}(\tub_1)\leq u'\leq u}\Big\{\lmu^{-\half}\big(|\Y\Lxi^{j+2}\psi|^2+\mu^2|\pv\Lxi^{j+2}\psi|^2\big)\notag\\
    &\qquad\qquad\qquad
    +\sum_{k\leq1}\Big(\lmu^{-\half}\big(|\Y\Lxi^j(\Carter^k\psi)_{\ell\geq 1}|^2+\mu^2|\pv\Lxi^j(\Carter^k\psi)_{\ell\geq 1}|^2\big)\Big)\Big\}\dSp\di u'\biggr)^\half\notag\\
    \lesssim {}&\tub_1^{\frac{3}{4}\gamma}
    \biggl( \iint\limits_{\tub=\tub_1, u_{\rblue}(\tub_1)\leq u'\leq u}\Big\{\sum_{k\leq1}\Enw^{(-\half)}[\Lxi^j(\Carter^k\psi)_{\ell\geq 1}]
    +\Enw^{(-\half)}[\Lxi^{j+2}\psi]\Big\}\dSp\di\ub'\biggr)^\half\notag\\
    \lesssim {}&  \tub_1^{\frac{3}{4}\gamma} \tub_1^{-4-j-\delta+\half\gamma}
    \lesssim {}\tub_1^{-4-j-\delta+\frac{5}{4}\gamma}\lesssim \tub_1^{-4-j-\epsilon }
  \end{align}
  for $\frac{5}{4}\gamma\leq \delta-\epsilon$,   where we have used
  \begin{align}
    \int_{u_{\rblue}(\tub_1)}^u\lmu^{\half}(u',\tub_1+r)\di u'
\lesssim \tub_1^{\half\gamma}(u-u_{\rblue}(\tub_1))\lesssim \tub_1^{\frac{3}{2}\gamma}
  \end{align}
  and  \eqref{SI:estimate:sph}
  to the second inequality, 
and \eqref{energy:decacy:full:blue:w} and \eqref{energy:decay:blue:higher:ub} to the third inequality.  This together with \eqref{eq:ellgeq1:highderi:I} then gives
  \begin{align}\label{pointwise:decay:highermode:blue:away}
    |\Lxi^j\Phi^{k_1}\Ckpsilarge(u,\ub, \phiout,\th)|\lesssim \ub^{-4-j-\epsilon}\ \ \text{in}\ \IIf.
  \end{align}
  
  In summary, in view of \eqref{pointwise:decay:psiz:blue:away} and \eqref{pointwise:decay:highermode:blue:away}, we conclude 
  \begin{align}\label{pointwise:decay:scalar:blue:away}
    |\Lxi^j\psi-c_0\Lxi^j(\ub^{-3}))|\lesssim \ub^{-3-j-\epsilon}\ \ \text{in}\ \IIf.
  \end{align}

  {\it Step 2. Computing the asymptotics of $\pv\psi$ in $\IIf$.}
Same to \eqref{eq:wave:ell=0:step2:I}, we rewrite the equation of $\psiz$ as
  \begin{align}
    \label{blue:eq:wave:ell=0:step1:I}
    \hspace{4ex}&\hspace{-4ex}
    \pv\biggl((\R) \pu\psiz+\half(\R)\Lxi\psiz\biggr)\notag\\
      ={}&\half(\R)\pv\Lxi\psiz-\frac{1}{4}a^2\mu \PJ_{\ell=0} (\sin^2\theta \Lxi^2\psi).
    \end{align}
Then $(\R) \pu\psiz+\half(\R)\Lxi\psiz$ can be calculated through integrating the  above equation from $\{r=\rblue\}$. The integral of $\pv\Lxi\psiz$ is bounded by $\ub^{-\frac{9}{2}+\half\gamma}\lesssim \ub^{-4-\epsilon}$ for $\gamma$ small enough, and the term $\mu\PJ_{\ell=0} (\sin^2\theta \Lxi^2\psi)$ is bounded by $\ub^{-5}$ from the pointwise decay estimate \eqref{pointwise:decay:scalar:blue:away}, hence the absolute value of the integral of the right-hand side of \eqref{blue:eq:wave:ell=0:step1:I} is bounded by $\ub^{-4-\epsilon}$. Using the precise asymptotics estimates \eqref{proof:I:ellz:872} for $\Lxi\psiz$ on $r=\rblue$ and \eqref{pointwise:decay:psiz:blue:away} for $\Lxi \psilarge$ in $\IIf$,  we arrive at
\begin{align}
  \bigg|\pu\Lxi^j\psiz(u,\ub)-\frac{3}{2}c_0\bigg(1-\frac{r_+^2+a^2}{\R}\bigg)\Lxi^j(\ub^{-4})\bigg|\lesssim \ub^{-4-j-\epsilon}\ \ \text{in}\ \IIf,
\end{align}
where the general $j\geq 0$ case follows since $\Lxi$ commutes with the wave operator. Here, we have also used on $r=\rblue$ that $\pu \Lxi^j \psiz =(\pv-\Lxi) \Lxi^j (\psiz) = \frac{3}{2}c_0(1-\frac{r_+^2+a^2}{\rblue^2+a^2})\Lxi^j(\ub^{-4})+O(\ub^{-4-j-\epsilon})$ from \eqref{proof:I:ellz:872} and \eqref{eq:precise:pvpsi:ell=0:I}.

Substituting \eqref{pointwise:decay:scalar:blue:away} with $j=1$ to the above inequality, we get
\begin{align}
  \label{eq:precise:pvpsi:ell=0:IIf}
    \bigg|\Lxi^j\big(\pv\psiz(u,\ub)\big)+\frac{3}{2}c_0\bigg(1+\frac{r_+^2+a^2}{\R}\bigg)\Lxi^j(\ub^{-4})\bigg|\lesssim \ub^{-4-\epsilon} \ \ \text{in}\ \IIf.
  \end{align}

Similarly, we integrate $\pu(\pv\psih\big|_{\Cur})=(\Y\pv\psih+\half\mu\pv^2\psih)\big|_{\Cur}$ along $\Cur$ in $\IIf$ from $\{r=\rblue\}$. Recall the wave equation satisfied by $\psih$, i.e.,
\begin{align}
\label{wave:high:nearCauchy:49}
  4(r^2+a^2)\Y \pv\psih+\mu(\Carter\psi)_{\ell\geq 1}
  +2a\mu\Lxi\Leta\psih
  +2\mu r(\Y\psih+ \pv\psih)
  +\frac{2ar\mu}{r^2+a^2}\Leta\psih=0.
\end{align}
This equation can be used to estimate the integral with respect to $\Y\pv\psih$ as follows:
\begin{align}\label{poof:2:2345}
  &\bigg|\int_{u_{\rblue}(\tub_1)}^{u} \Y\pv\psih\big|_{\Cur(u')}\di u'\bigg|\notag\\
  \lesssim {}&\biggl(\int_{u_{\rblue}(\tub_1)}^{u}-\mu\di u'\biggr)^\half\biggl(\int_{u_{\rblue}(\tub_1)}^{u}-\mu\big(|\Carter\psih|^2+|\Lxi^{\leq 1}\Leta\psih|^2+|\Y\psih|^2+|\pv\psih|^2)\big|_{\Cur}\di u'\biggr)^\half\notag\\
  \lesssim {}&(\rblue-r)^\half \biggl(\iint\limits_{\tub=\tub_1, u_{\rblue}(\tub_1)\leq u'\leq u}\big(\sum_{j\leq 1}\Enw^{(-\half)}[\Lxi^j(\Carter^{\leq 2}\psi)_{\ell\geq 1}]+\Enw^{(-\half)}[\Carter^{\leq 1}\Lxi\psi]\big)\dSp \di u'\biggr)^\half\notag\\
  \lesssim {}&\tub_1^{-4-\delta+\half \gamma}\lesssim \tub_1^{-4-\epsilon}\ \ \text{in}\ \IIf,
\end{align}
for $\half \gamma\leq \delta-\epsilon$, 
where we have used \eqref{SI:estimate:sph} for the second step and \eqref{energy:decacy:full:blue:w} and \eqref{energy:decay:blue:higher:ub} for the third step. On the other hand, $\mu\pv^2\psih$ can be expressed as $\mu(\Y\pv\psih+\frac{a}{\R}\pv\Leta\psih+\pv\Lxi\psih)$, hence one finds the integral of $\mu\pv^2\psih$ along  $\Cur$ is bounded by $\tub_1^{-4-\epsilon}$, which implies that the integral of $\pv\psih\big|_{\Cur}$  along $\Cur$ in $\IIf$ is bounded by $\ub^{-4-\eps}$. Further, since the absolute value of $\pv\psih$ on $\{r=\rblue\}$ is bounded by $\ub^{-4-\eps}$ from the estimate \eqref{eq:precise:e4psilarge:I:739}, we achieve $ |\pv\psih(u,\ub,\phiin,\th)|\lesssim \ub^{-4-\epsilon} \ \ \text{in}\ \IIf$. Commuting with $\Lxi$, we generalize it to the following estimate for any $j\geq 0$:
\begin{align}\label{pointwise:decay:higher:II}
  |\pv\Lxi^j\psih(u,\ub,\phiin,\th)|\lesssim \ub^{-4-j-\epsilon} \ \ \text{in}\ \IIf.
\end{align}

In the end, combining \eqref{eq:precise:pvpsi:ell=0:IIf} and \eqref{pointwise:decay:higher:II} together, and due to $e_4=\ptr_{\ub}$, we deduce
\begin{align}
  \label{eq:precise:pvpsi:scalar:IIf}
    \bigg|\Lxi^j\big(e_4\psi(u,\ub,\phiout,\th)\big)+\frac{3}{2}c_0\bigg(1+\frac{r_+^2+a^2}{\R}\bigg)\Lxi^j(\ub^{-4})\bigg|\lesssim \ub^{-4-j-\epsilon} \ \ \text{in}\ \IIf.
  \end{align}
  Again, the general $j\geq 0$ case is manifest since $\Lxi$ commutes with the wave operator.
  Meanwhile, by $\Y=\pv-\Lxi -\frac{a}{\R}\Leta$, and by $\abs{\Leta\psi}=\abs{\Leta(\psi_{\ell\geq 1})}\lesssim \ub^{-4-\epsilon}$ which follows from \eqref{pointwise:decay:highermode:blue:away}, one finds
\begin{align}
  \label{eq:precise:pupsi:full:IIf}
    \bigg|\Lxi^j\big(\Y\psi(u,\ub,\phiout,\th)\big)-\frac{3}{2}c_0\bigg(1-\frac{r_+^2+a^2}{r^2+a^2}\bigg)\Lxi^j(\ub^{-4})\bigg|\lesssim \ub^{-4-j-\epsilon}\ \ \text{in}\ \IIf.
\end{align}

\subsubsection{Estimates in region $\II$}
\label{sect:IImoduloIIf}

  {\it Step 1. Computing the asymptotics of $\pv\psi$.}
  Integrating equation
  \begin{align}
    \pu\bigl((\R)\pv\psiz\bigl)
  =\half \mu r\Lxi\psiz-\frac{1}{4}a^2\mu\PJ_{\ell=0} (\sin^2\theta \Lxi^2\psi)
  \end{align}
  along $\ub=\ub_1$ from $\Gamma$, we can show 
  \begin{align}\label{proof:3:decay:psiz}
    |(\R)\pv\psiz(u,\ub_1)-(\R)\pv\psiz(u_{\Gamma}(\ub_1),\ub_1)|\lesssim e^{-(1-\epsilon')|\kappa_-|\ub_1^\gamma}.
  \end{align}
  In fact, the error terms are controlled as follows:
  \begin{subequations}
  \begin{align}
    &\biggl|\int_{u_{\Gamma}(\ub)}^u\mu r\Lxi\psiz(u',\ub)\di u'\biggr|\notag\\
    \lesssim {}&\biggl(\int_{u_{\Gamma}(\ub)}^u\mu^2\lmu^{\frac{5}{2}}(u',\ub)\di u'\biggr)^\half\biggl(\int_{u_{\Gamma}(\ub)}^u\lmu^{-\frac{5}{2}}|\Lxi\psiz|^2(u',\ub)\di u'\biggr)^\half\notag\\
    \lesssim {}&e^{-(1-\epsilon')|\kappa_-|\ub^\gamma},\\
    &\biggl|\int_{u_{\Gamma}(\ub)}^u\mu \PJ_{\ell=0}(\sin^2\th \Lxi^2\psi)(u',\ub)\di u'\biggr|\notag\\
    \lesssim {}&\biggl(\int_{u_{\Gamma}(\ub)}^u\mu^2\lmu^{\frac{5}{2}}(u',\ub)\di u'\biggr)^\half \times \biggl(\int_{u_{\Gamma}(\ub)}^u\lmu^{-\frac{5}{2}}|\PJ_{\ell=0}(\sin^2\th \Lxi^2\psi)|^2(u',\ub)\di u'\biggr)^\half\notag\\
    \lesssim {}&e^{-(1-\epsilon')|\kappa_-|\ub^\gamma}\biggl(\iint\limits_{\ub=\ub_1, u_{\Gamma}(\ub_1)\leq u\leq u_1}\lmu^{-\frac{5}{2}}|\sin^2\th \Lxi^2\psi|^2\dSp\di u\biggr)^\half\notag\\
    \lesssim {}&e^{-(1-\epsilon')|\kappa_-|\ub^\gamma}\biggl(\iint\limits_{\ub=\ub_1, u_{\Gamma}(\ub_1)\leq u\leq u_1}\lmu^{-\frac{5}{2}}| \Lxi^2\psi|^2\dSp\di u\biggr)^\half\notag\\
    \lesssim {}&e^{-(1-\epsilon')|\kappa_-|\ub^\gamma},
  \end{align}
  \end{subequations}
  where we have used \eqref{energy:decay:blue:psiz:nearHor}, \eqref{energy:decay:blue:psi:nearHor} and 
  \begin{align}
    \int_{u_{\Gamma}(\ub)}^u\mu^2\lmu^{\frac{5}{2}}(u',\ub)\di u' \lesssim & \int_{r}^{r(u_{\Gamma}(\ub),\ub)}(r-r_-)|\log(r-r_-)|^{\frac{5}{2}}\di r\notag\\
    \lesssim & |\mu(u_{\Gamma}(\ub),\ub)|^{2-2\epsilon'}\lesssim e^{-(2-2\epsilon')|\kappa_-|\ub^\gamma},
  \end{align}
  for any $1>2\epsilon'>0$.

 Substituting \eqref{eq:precise:pvpsi:ell=0:IIf} to \eqref{proof:3:decay:psiz}, and using 
 \begin{align}
  |r^2(u,\ub)-r^2(u_{\gamma}(\ub),\ub)|\ub^{-4}\lesssim -\mu(u_{\gamma}(\ub),\ub) \ub^{-4}\lesssim e^{-|\kappa_-|\ub^\gamma}\ub^{-4},
 \end{align} 
  we obtain
 \begin{align}
  \label{eq:precise:pvpsi:ell=0:nearHor}
    \bigg|\pv\psiz(u,\ub)+\frac{3}{2}c_0\bigg(1+\frac{r_+^2+a^2}{\R}\bigg)\ub^{-4}\bigg|\lesssim \ub^{-4-\epsilon} \ \ \text{in}\ \II\backslash \IIf.
  \end{align}

Similarly,  for higher modes $\psih$, we integrate $\pu(\pv\psih\big|_{\Cur})=(\Y\pv\psih+\half\mu\pv^2\psih)\big|_{\Cur}$ along $\Cur$ from $\Gamma$. One can obtain from \eqref{wave:high:nearCauchy:49} that 
\begin{align}
  \bigg|\int_{u_{\Gamma}(\tub_1)}^{u} \Y\pv\psih\big|_{\Cur(u')}\di u'\bigg| \lesssim (r(u_{\Gamma}(\tub_1),\tub_1)-r)^\half\ \lesssim |\mu (u_{\Gamma}(\tub_1),\tub_1)|^\half\lesssim e^{-\half|\kappa_-|\ub^\gamma},
\end{align}
which  can be proved in a similar way as proving \eqref{poof:2:2345} by using the boundedness estimates \eqref{energy:bound:blue:higher:ub} and \eqref{energy:decacy:full:blue:v}. Further, the integral of $\mu\pv^2\psih$ is bounded by $e^{-\half|\kappa_-|\ub^\gamma}$. Combined with the estimate \eqref{pointwise:decay:higher:II} of $\pv\psih=e_4\psih$ on $\Gamma$, we achieve for any $k\geq 0$ and $j\geq 0$,
\begin{align}\label{pointwise:decay:higher:II:nearHor}
  |\pv\Leta^k\Lxi^j\psih(u,\ub)|\lesssim \ub^{-4-j-\epsilon} \ \ \text{in}\ \II\backslash \IIf,
\end{align}
where the general case with $k\geq 0$ and $j\geq 0$ holds by the simple fact that $\Phi$ and $\Lxi$ commute with the wave operator.

In the end, we conclude from  \eqref{eq:precise:pvpsi:scalar:IIf},  \eqref{eq:precise:pvpsi:ell=0:nearHor} and \eqref{pointwise:decay:higher:II:nearHor}  that 
 \begin{align}
    \label{eq:precise:pvpsi:scalar:nearHor}
      \bigg|\Lxi^j\big(e_4\psi(u,\ub,\phiout,\th)\big)+\frac{3}{2}c_0\bigg(1+\frac{r_+^2+a^2}{r^2+a^2}\bigg)\Lxi^j(\ub^{-4})\bigg|\lesssim \ub^{-4-j-\epsilon} \ \ \text{in}\ \II.
    \end{align}
    
    By the same argument in region $\II'$, we obtain
 \begin{align}
    \label{eq:precise:pvpsi:scalar:nearHor:LHS}
      \bigg|\Lxi^j\big(e_3'\psi(u,\ub,\phiin,\th)\big)+\frac{3}{2}c_0'\bigg(1+\frac{r_+^2+a^2}{r^2+a^2}\bigg)\Lxi^j(u^{-4})\bigg|\lesssim u^{-4-j-\epsilon} \ \ \text{in}\ \II'.
    \end{align}

  {\it Step 2. Stability on Cauchy horizon.}
Since the decay rate of $\pv\psi(u,\ub,\phiout,\th)$ in \eqref{eq:precise:pvpsi:scalar:nearHor} is integrable, we conclude
  \begin{align}\label{0000}
    \Psi(u,\phiout,\th)\doteq \lim\limits_{\ub\to +\infty}\psi(u,\ub,\phiout,\th) \ \ \text{exists and is finite}.
  \end{align}
  Integrating \eqref{eq:precise:pvpsi:scalar:nearHor} with respect to $\di \ub$, and using the following inequality
  \begin{align}
    &\biggl|\int_{\ub}^{+\infty}\bigg(1+\frac{r_+^2+a^2}{\R}\bigg)(\ub')^{-4}\di \ub'-\bigg(1+\frac{r_+^2+a^2}{r_-^2+a^2}\bigg)\int_{\ub}^{+\infty}(\ub')^{-4}\di \ub'\biggr|\notag\\
   \lesssim &\int_{\ub}^{+\infty}(-\mu)(\ub')^{-4}\di \ub'\lesssim -\mu(u,\ub)\ub^{-4},
  \end{align}
  we get 
  \begin{align}\label{stability:pointwise:scalar}
    \biggl|\psi(u,\ub,\phiout,\th)-\Psi(u,\phiout,\th)-\half c_0\bigg(1+\frac{r_+^2+a^2}{r^2+a^2}\bigg)\ub^{-3}\biggr|\lesssim \ub^{-3-\epsilon} \ \ \text{in}\ \II.
  \end{align}
  Similar to \eqref{0000}, we can also conclude that $\lim\limits_{\ub\to\infty}T\psi(u,\ub,\phiout,\th)$ exists. In view that  $\partial_{\ub}\psi$ converges to $0$ uniformly, and using $T=\partial_{\ub}-\partial_u$, this yields the uniform convergence of $\pu \psi$ as $\ub\to +\infty$. Together with the uniform convergence of $\psi$ to $\Psi$ as $\ub\to +\infty$, we obtain from the above that
   \begin{align}
  \lim_{\ub\to+\infty}\pu\psi(u,\ub,\phiout,\th)=\pu\Psi(u,\phiout,\th)=-\Lxi\Psi(u,\phiout,\th)\  \text{in}\ \II.
  \end{align}
  Combining this with \eqref{eq:precise:pvpsi:scalar:nearHor}, and in view that $\Lxi=\pv-\pu$, we obtain
  \begin{align}\label{stability:pointwise:scalar:Lxi}
    \biggl|\Lxi^j\psi(u,\ub,\phiout,\th)-\Lxi^j\Psi(u,\phiout,\th)-\half c_0\bigg(1+\frac{r_+^2+a^2}{r^2+a^2}\bigg)\Lxi^j(\ub^{-3})\biggr|\lesssim \ub^{-3-j-\epsilon} \ \ \text{in}\ \II .
  \end{align}
Here, the general $j\geq 0$ follows by an induction.

  {\it Step 3. Computing the asymptotics of $\Psi(u,\phiout,\th)$ for $u\leq 1$.} Recall the estimate \eqref{pointwise:decay:scalar:blue:away} on $\Gamma$:
  \begin{align}
    |\Lxi^j\psi(u_{\Gamma}(\ub),\ub,\phiout,\th)-c_0\Lxi^j(\ub^{-3})|\lesssim \ub^{-3-j-\epsilon}.
  \end{align}
  Substituting this inequality back into \eqref{stability:pointwise:scalar:Lxi}, we obtain
  \begin{align}
    \biggl|\Lxi^j\Psi(u_{\Gamma}(\ub),\phiout,\th)-\half c_0\bigg(1-\frac{r_+^2+a^2}{r_-^2+a^2}\bigg)\Lxi^j(\ub^{-3})\biggr|\lesssim \ub^{-3-j-\epsilon} .
  \end{align}
  The definition of the hypersurface $\Gamma$, on which $u_{\Gamma}(\ub)+\ub=2r^*(u_{\Gamma}(\ub),\ub)=\ub^{\gamma}$, implies $\frac{u_{\Gamma}(\ub)}{\ub}=-1+\ub^{\gamma-1}$. Thus, we get
  \begin{align}
  \label{eq:Lxij:Psi:NearCH:unega}
    \biggl|\Lxi^j\Psi(u,\phiout,\th) + \half c_0\bigg(1-\frac{r_+^2+a^2}{r_-^2+a^2}\bigg)\Lxi^j({u}^{-3})\biggr|\lesssim \abs{u}^{-3-j-\epsilon}\ \ \text{as}\ u\to -\infty.
  \end{align}

  {\it Step 4. Computing the asymptotics of $\Y\psi$.}
By rewriting \eqref{eq-wave-doublenull-2} as
  \begin{align}
    4(r^2+a^2)^\half \pv\big((\R)^\half\Y\psi\big)+\mu\bigg(\Carter\psi
  +2a\Lxi\Leta\psi
  +2 r\pv\psi
  -\frac{2ar}{r^2+a^2}\Leta\psi\bigg)=0,
  \end{align}
and using the boundedness of $\Carter\psi$, $\Lxi^{\leq 1}\Leta\psi$ and $\pv\psi$ and the fact that $r^2+a^2 = r_{-}^2 +a^2 + O(\abs{\mu})$, the limit $\Psi_1(u,\phiout, \th):=\lim_{\ub\to +\infty} \Y\psi(u,\ub,\phiout,\th)$ exists and it holds
  \begin{align}
  \label{eq:limit:Ypsi:closetoCHorizon}
   \bigl|\Y\psi(u,\ub,\phiout,\th)-\Psi_1(u,\phiout,\th)\bigr|\lesssim -\mu(u,\ub)\ \ \text{in}\ \II\backslash \IIf.
   \end{align}
    Recall that $\Y=\pu-\frac{a}{\R}\Leta=\pu-\frac{a}{r_-^2+a^2}\Leta +O(\abs{\mu}) \Leta=\Y\vert_{\CHorizon} +O(\abs{\mu}) \Leta$. This yields $|\Y\vert_{\CHorizon}\psi(u,\ub,\phiout,\th)-\Psi_1(u,\phiout,\th)\bigr|\lesssim -\mu(u,\ub)$ in $\II\backslash \IIf$, and thus $\Psi_1= \Y\vert_{\CHorizon}(\Psi)$. Plugging this back to \eqref{eq:limit:Ypsi:closetoCHorizon}, we achieve $\bigl|\Y\psi(u,\ub,\phiout,\th)-\Y\vert_{\CHorizon}(\Psi(u,\phiout,\th))\bigr|\lesssim -\mu(u,\ub)$ in $\II\backslash \IIf$. Notice that this argument works for $\II$ as well, hence, by recalling $\Y=-\mu e_3$, we obtain
    \begin{align}\label{pointwise:decay:Ypsi:nearHor}
    \bigl|\Lxi^j\big((-\mu e_3)\psi(u,\ub,\phiout,\th)\big)-\Lxi^j\big((-\mu e_3)\vert_{\CHorizon}(\Psi(u,\phiout,\th))\big)\bigr|\lesssim -\mu(u,\ub)\ \ \text{in}\ \II,
  \end{align}
  where the general $j\geq 0$ case is trivial by commuting the wave equation with $\Lxi^j$.
  
  Similar to the argument in Step 3, we have from \eqref{eq:precise:pupsi:full:IIf} that
 \begin{align}
 \big|\Lxi^j\big(\hat{e}_3 \psi (u_{\Gamma}(\ub),\phiout,\th)\big)-\frac{3}{2}c_0\frac{r_-^2-r_+^2}{r_{-}^2 +a^2}\Lxi^j\big(\ub^{-4}\big)\big|\lesssim \ub^{-4-j-\epsilon},
 \end{align}
  and thus, $\big|\Lxi^j\big(\Y\vert_{\CHorizon}(\Psi(u,\phiout,\th))\big)-\frac{3}{2}c_0(1-\frac{r_+^2+a^2}{r_-^2+a^2})\Lxi^j\big(u^{-4}\big)\big|\lesssim u^{-4-j-\epsilon}$, which then yields
 \begin{align}
 \label{eq:YPsi:CH:precise:unega}
 \bigg|\Lxi^j\big((-\mu e_3)\vert_{\CHorizon}(\Psi(u,\phiout,\th))\big)-\frac{3}{2}c_0\bigg(1-\frac{r_+^2+a^2}{r_-^2+a^2}\bigg)\Lxi^j\big({u}^{-4}\big)\bigg|\lesssim \abs{u}^{-4-j-\epsilon} \ \ \text{as}\ u\to -\infty.
 \end{align}

{\it Step 5. Computing the asymptotics of $\Psi(u,\phiout,\th)$ in $u\geq 1$.} 
In view of  the two  estimates  \eqref{eq:precise:pvpsi:scalar:nearHor:LHS}  and \eqref{pointwise:decay:Ypsi:nearHor} in $\II\cap \II'$, it actually holds that
	\begin{align}
		\label{eq:Lxij:Psi:NearCH:uposi}
		\biggl|\Lxi^j\big((-\mu e_3)\vert_{\CHorizon}(\Psi(u,\phiout,\th))\big)+\frac{3}{2}c_0'\bigg(1+\frac{r_+^2+a^2}{r_-^2+a^2}\bigg)\Lxi^j(u^{-4})\biggr|\lesssim {}& u^{-4-j-\epsilon} \ \ \text{as}\ u\to +\infty,
	\end{align}
where we recall that $e_3'=-\mu e_3$. Integrating the above inequality along the integral curve  $\gamma_{(\phiout,\th)}(u)$ of $(-\mu e_3)\vert_{\CHorizon}=\pu-\frac{a}{r_-^2+a^2}\partial_{\phiout}$ which starts from $\Sphere_{1,+\infty}$, we achieve
\begin{align}\label{oooo}
	\lim\limits_{u\to+\infty}\Lxi^j\Psi(\gamma_{(\phiout,\th)}(u))=\Upsilon_j(\phiout, \th),
\end{align}
for some function $\Upsilon_j$ defined on $\Sphere$, and 
\begin{align}\label{asy:Psi:positiveu}
	\Biggl|\Lxi^j\Psi(\gamma_{(\phiout,\th)}(u))
	-\Upsilon_j(\phiout, \th)-\frac{1}{2}c_0'\bigg(1+\frac{r_+^2+a^2}{r_-^2+a^2}\bigg)\Lxi^j(u^{-3})\biggr|\lesssim u^{-3-j-\epsilon}\ \ \text{as}\ u\to+\infty.
\end{align}
Similar to \eqref{stability:pointwise:scalar:Lxi} and \eqref{asy:Psi:positiveu}, we have
\begin{align}\label{stability:pointwise:scalar:Lxi:leftH}
	\biggl|\Lxi^j\psi(u,\ub,\phiin,\th)-\Lxi^j\Psi'(\ub,\phiin,\th)-\half c_0'\bigg(1+\frac{r_+^2+a^2}{r^2+a^2}\bigg)\Lxi^j(u^{-3})\biggr|\lesssim u^{-3-j-\epsilon} \ \ \text{in}\ \II' ,
\end{align}
and
\begin{align}\label{asy:Psi:positivev:left}
	\Biggl|\Lxi^j\Psi'(\gamma_{(\phiin,\th)}(\ub))-\Upsilon_j'(\phiin, \th)-\frac{1}{2}c_0\bigg(1+\frac{r_+^2+a^2}{r_-^2+a^2}\bigg)\Lxi^j(\ub^{-3})\biggr|\lesssim \ub^{-3-j-\epsilon}\ \ \text{as}\  \ub\to+\infty,
\end{align}
for  some functions $\Psi'$ defined on $\CHorizon'$ and  $\Upsilon_j'$ defined on $\Sphere$.

  \subsection{Complete the proof of main Theorem \ref{thm:main}}
   \label{sect:pf:mainthm}

By the estimates \eqref{eq:main:ptw:I:psi}, \eqref{stability:pointwise:scalar:Lxi}, \eqref{eq:Lxij:Psi:NearCH:unega} and \eqref{eq:Lxij:Psi:NearCH:uposi}, for each $j\in \mathbb{N}$, there exist smooth functions $\Psi(u,\omega)$ and $\Upsilon_j(\omega)$,  $\omega$ being the spherical coordinates on $\Sp$, such that
 \begin{subequations}
 \label{eq:proof:Main:psi:23}
 \begin{align}
 \label{pointwise:decay:scalar:blue:away:23}
    |\Lxi^j\psi-c_0\Lxi^j(\ub^{-3})|\lesssim {}& \ub^{-3-j-\epsilon}\ \ \text{in}\ \I\cup\IIf,\\
\label{stability:pointwise:scalar:Lxi:23}
    \biggl|\Lxi^j\psi-\Lxi^j\Psi(u,\omega)-\half c_0\bigg(1+\frac{r_+^2+a^2}{r^2+a^2}\bigg)\Lxi^j(\ub^{-3})\biggr|\lesssim {}&\ub^{-3-j-\epsilon} \ \ \text{in}\ \II ,
  \end{align}
  where
   \begin{align}
  \label{eq:Lxij:Psi:NearCH:unega:23}
 \biggl|\Lxi^j\Psi(u,\omega) +\half c_0\bigg(1-\frac{r_+^2+a^2}{r_-^2+a^2}\bigg)\Lxi^j({u}^{-3})\biggr|\lesssim{}& \abs{u}^{-3-j-\epsilon}\ \ \text{as}\ u\to -\infty,\\
 \label{eq:Lxij:Psi:NearCH:uposi:sum}
    \biggl|\Lxi^j\Psi(\gamma_\omega(u))-\Upsilon_j(\omega)-\half c_0'\bigg(1+\frac{r_+^2+a^2}{r_-^2+a^2}\bigg)\Lxi^j({u}^{-3})\biggr|\lesssim {}& \abs{u}^{-3-j-\epsilon} \ \ \text{as}\ u\to +\infty,
  \end{align}
 \end{subequations}
with $\gamma_{\omega}(u)$ being the integral curve  of $(-\mu e_3)\vert_{\CHorizon}$ starting from $\Sphere_{1,+\infty}$. 
 These justify the estimates \eqref{eq:main:R:psi} for any choice of $\gamma_0\in (0,1)$.\footnote{Here, the choice of $\gamma_0$ is independent of the $\gamma$ in Subsection \ref{sect:IIf}.}  In fact, in $\II\cap\{2r^*\leq \ub^{\gamma_0}\}$, we have 
 \begin{align}
 \label{eq:coordrelas:IIGamma0}
 \bigg|1+\frac{u}{\ub}\bigg|\lesssim \ub^{\gamma_0-1}, 
 \end{align}
  thus, applying this relation,  one finds the inequality \eqref{pointwise:decay:scalar:blue:away:23}  holds true in $\II\cap \{ 2r^*\leq \ub^{\gamma_0}\}$ by combining the estimates  \eqref{stability:pointwise:scalar:Lxi:23} and \eqref{eq:Lxij:Psi:NearCH:unega:23}.

 For $e_4\psi$, by the estimate \eqref{eq:main:ptw:I:pubpsi}  in region $\I$ and the estimates \eqref{eq:precise:pvpsi:scalar:IIf} and \eqref{eq:precise:pvpsi:scalar:nearHor} in region $\II$, we conclude
 the following precise late-time asymptotics for $e_4\psi$:
 \begin{align}
    \label{eq:precise:pvpsi:scalar:RHS:23}
      \bigg|\Lxi^j(e_4\psi)+\frac{3}{2}c_0\bigg(1+\frac{r_+^2+a^2}{r^2+a^2}\bigg)\Lxi^j(\ub^{-4})\bigg|\lesssim \ub^{-4-j-\epsilon} \ \ \text{in}\ \DD\cap \{\ub\geq 1\},
    \end{align}
    which is exactly \eqref{eq:main:R:e4psi}.

For $(-\mu e_3)\psi$, we conclude from the estimates \eqref{eq:main:ptw:I:pupsi}, \eqref{eq:precise:pupsi:full:IIf}, \eqref{pointwise:decay:Ypsi:nearHor} and \eqref{eq:YPsi:CH:precise:unega}  the following precise late-time asymptotics:
\begin{subequations}
\begin{align}
  \label{eq:precise:pupsi:full:IIf:23}
    \bigg|\Lxi^j\big((-\mu e_3)\psi\big)-\frac{3}{2}c_0\bigg(1-\frac{r_+^2+a^2}{r^2+a^2}\bigg)\Lxi^j(\ub^{-4})\bigg|\lesssim (r_+-r) \ub^{-4-j-\epsilon}\ \ \text{in}\ \I\cup \IIf,
\end{align}
  \begin{align}\label{pointwise:decay:Ypsi:nearHor:23}
    \bigl|\Lxi^j\big((-\mu e_3)\psi\big)-\Lxi^j\big((-\mu e_3)\vert_{\CHorizon}(\Psi(u,\omega))\big)\bigr|\lesssim -\mu\ \ \text{in}\  \II,
  \end{align}
  where
\begin{align}
\label{eq:YPsi:CH:precise:unega:23}
\bigg|\Lxi^j\big((-\mu e_3)\vert_{\CHorizon}(\Psi(u,\omega))\big)-\frac{3}{2}c_0\bigg(1-\frac{r_+^2+a^2}{r_-^2+a^2}\bigg)\Lxi^j\big({u}^{-4}\big)\bigg|\lesssim \abs{u}^{-4-j-\epsilon} \ \ \text{as}\ u\to -\infty.
\end{align}
\end{subequations}
These prove the estimates \eqref{eq:main:R:e3psi}  for any choice of $\gamma_0\in (0,1)$ by the same argument above via using the relation \eqref{eq:coordrelas:IIGamma0}.

\section*{Acknowledgement}

The first author S. M. acknowledges the support by the ERC grant ERC-2016 CoG 725589 EPGR and the Alexander von Humboldt postdoc fellowship. The second author L. Z.  acknowledges the support by the National Nature Science Foundation of China (Grant No. 12201083). The authors are grateful to the anonymous referees for many valuable comments and suggestions.


\bibliographystyle{amsplain}
\newcommand{\arxivref}[1]{\href{http://www.arxiv.org/abs/#1}{{arXiv.org:#1}}}
\newcommand{\mnras}{Monthly Notices of the Royal Astronomical Society}
\newcommand{\prd}{Phys. Rev. D }
\newcommand{\prl}{Phys. Rev. Lett. }
\newcommand{\apj}{Astrophysical J. }

\providecommand{\MR}{\relax\ifhmode\unskip\space\fi MR }
\providecommand{\MRhref}[2]{%
  \href{http://www.ams.org/mathscinet-getitem?mr=#1}{#2}
}
\providecommand{\href}[2]{#2}

\end{document}